\documentclass{article}

\usepackage{amssymb,amsthm,latexsym,amsmath,graphics,setspace}
\usepackage{epsfig,epsf,amsfonts}
\usepackage{color,multicol,colordvi,graphicx,multirow}
\usepackage{natbib,fancyhdr}
\newtheorem{theorem}{Theorem}[section]
\newtheorem{lemma}[theorem]{Lemma}

\newtheorem{proposition}[theorem]{Proposition}
\newtheorem{definition}[theorem]{Definition}

\newcommand{\paren}[1]{\left(#1\right)}
\newcommand{\D}[2]{\frac{d#1}{d#2}}
\newcommand{\DD}[3]{\frac{d^{#1}{#2}}{d{#3}^{#1}}}
\newcommand{\PD}[2]{\frac{\partial#1}{\partial#2}}

\newcommand{\at}[2]{\left. #1 \right|_{#2}}

\newcommand{\mb}[1]{\mathbf{#1}}
\newcommand{\bm}[1]{\boldsymbol{#1}}
\newcommand{\abs}[1]{\left\lvert #1 \right\rvert}

\newcommand{\ip}[2]{\left\langle #1,#2 \right\rangle}
\newcommand{\wh}[1]{\widehat{#1}}
\newcommand{\wt}[1]{\widetilde{#1}}

\numberwithin{equation}{section}

\author{Yoichiro Mori\\
School of Mathematics\\
University of Minnesota\\
206 Church St SE, Minneapolis MN, 55455}
\title{Mathematical Properties 
of Pump-Leak Models of Cell Volume Control and Electrolyte Balance}
\date{Oct. 9, 2011}

\begin{document}

\maketitle

\begin{abstract}
Homeostatic control of cell volume and intracellular electrolyte content is a fundamental problem in physiology and is central to the functioning of epithelial systems. These physiological processes are modeled using pump-leak models, a system of differential algebraic equations that describes the balance of ions and water flowing across the cell membrane. Despite their widespread use, very little is known about their mathematical properties. Here, we establish analytical results on the existence and stability of steady states for a general class of pump-leak models. We treat two cases. When the ion channel currents have a linear current-voltage relationship, we show that there is at most one steady state, and that the steady state is globally asymptotically stable. If there are no steady states, the cell volume tends to infinity with time. When minimal assumptions are placed on the properties of ion channel currents, we show that there is an asymptotically stable steady state so long as the pump current is not too large. The key analytical tool is a free energy relation satisfied by a general class of pump-leak models, which can be used as a Lyapunov function to study stability.\newline
{\em Key Words}: Cell Volume Control, Electrolyte Balance, Free Energy,
Lyapunov Function, Differential Algebraic System.
\end{abstract}

\section{Introduction}

Cells contain a large number of organic molecules
that do not leak out through the cell membrane. 
The presence of organic molecules and 
their attendant counterions results in excess intracellular 
osmotic pressure. The plasma membrane is not mechanically strong 
enough to withstand significant differences in osmotic pressure,
and thus the cell will tend to swell and burst. 
Plant cells and bacteria have a mechanically rigid 
cell wall to guard against this tendency.
Animal cells maintain their cell volume with ionic pumps and 
ionic channels that together regulate the ionic composition 
of the cytosol \citep{hoffmann2009physiology,evans2009osmotic,boron2008medical}. 
 
This ``pump-leak'' mechanism is typically modeled in the following fashion
\citep{KS,hoppensteadt2002modeling}.
Consider a cell of volume $v$ and let 
$[\cdot]_{\rm i,e}$ be the intracellular and extracellular ionic concentrations
respectively. We only consider the ions Na$^+$, K$^+$ and Cl$^-$.
Let the cell be in a large and well-stirred extracellular 
bath so that $[\cdot]_{\rm e}$ can be assumed constant.
We have the following balance equations for the three ionic species.
\begin{subequations}\label{NaKCl}
\begin{align}
\D{}{t}\paren{Fv[\text{Na}^+]_{\rm i}}
&=-g_{\rm Na}\paren{\phi-\frac{RT}{F}\ln \paren{\frac{[\text{Na}^+]_{\rm e}}{[\text{Na}^+]_{\rm i}}}}-3\alpha F,\label{Na}\\
\D{}{t}\paren{Fv[\text{K}^+]_{\rm i}}
&=-g_{\rm K}\paren{\phi-\frac{RT}{F}\ln \paren{\frac{[\text{K}^+]_{\rm e}}{[\text{K}^+]_{\rm i}}}}+2\alpha F,\label{K}\\
\D{}{t}\paren{-Fv[\text{Cl}^-]_{\rm i}}
&=-g_{\rm Cl}\paren{\phi+\frac{RT}{F}\ln \paren{\frac{[\text{Cl}^-]_{\rm e}}{[\text{Cl}^-]_{\rm i}}}}.\label{Cl}
\end{align}
Here, $\phi$ is the membrane potential,  
$g_{\cdot}$ are the ion channel conductances for each species of ion, 
$\alpha$ is the strength 
of the pump current, $F$ is the Faraday constant, and 
$RT$ is the ideal gas constant times absolute temperature.
The pump current for Na$^+$ and K$^+$ has a ratio of $3:2$ reflecting
the stoichiometry of the Na-K ATPase (in \citep{hoppensteadt2002modeling}, 
this ratio is set to $1:1$ for simplicity).
The above balance laws are supplemented by the following:
\begin{align}
\nonumber 0=&[\text{Na}^+]_{\rm i}+[\text{K}^+]_{\rm i}-[\text{Cl}^{-}]_{\rm i}+\frac{zA}{v}\\
=&[\text{Na}^+]_{\rm e}+[\text{K}^+]_{\rm e}-[\text{Cl}^{-}]_{\rm e},\label{cmv}\\
\nonumber \D{v}{t}
=&\zeta RT\Bigl([\text{Na}^+]_{\rm i}+[\text{K}^+]_{\rm i}+[\text{Cl}^{-}]_{\rm i}+\frac{A}{v}\\
&-([\text{Na}^+]_{\rm e}+[\text{K}^+]_{\rm e}+[\text{Cl}^{-}]_{\rm e})\Bigr).\label{w}
\end{align}
\end{subequations}
Equation \eqref{cmv} is the electroneutrality condition, where 
$A$ is the total amount of organic molecules in the cell and $z$ the average 
charge on one organic molecule. We have assumed that there are no
organic molecules outside the cell and that they do not pass through 
the membrane. Equation \eqref{w} says that water 
flows into or out of the cell according to the osmotic pressure difference
across the membrane. Here, $\zeta$ is the membrane permeability to water flow.
System \eqref{NaKCl} forms a system of differential algebraic equations.
We would like to see under what condition the above system possesses 
a stable steady state, representing a cell with a stable cell volume 
and ionic composition.

Models of this type were first 
introduced in \citep{tosteson1960regulation,tosteson1964regulation} 
and have since been extended and modified by 
several authors to study cell volume control
\citep{jakobsson1980interactions,lew1991mathematical,
hernandez_modeling_1998,armstrong2003k}.
Pump-leak models are widely used in epithelial physiology.
Epithelial cells need to maintain their 
cell volume and ionic composition in the face of widely varying 
extracellular ionic and osmotic conditions. There is a large body 
of mathematical modeling work in the context of renal physiology 
(see \citep{weinstein1994mathematical,weinstein2003mathematical} 
for review). Mathematical modeling studies of other systems 
using pump-leak models include \citep{larsen2002analysis,
fischbarg2005mathematical,yi2003mathematical}.

Despite their widespread use and fundamental physiological importance, 
there seem to be very few analytical results 
regarding the behavior of pump-leak models. 
As for system \eqref{NaKCl}, \citep{KS} shows the following.
Assuming $z\leq -1$, there is a unique steady state with 
a finite positive cell volume if and only if:
\begin{equation}\label{NaKClcond}
\frac{[{\rm Na}^+]_{\rm e}\exp(-3\alpha F/(g_{\rm Na}RT/F))+
[{\rm K}^+]_{\rm e}\exp(2\alpha F/(g_{\rm K}RT/F))}
{[{\rm Na}^+]_{\rm e}+[{\rm K}^+]_{\rm e}}<1.
\end{equation}
This means that if condition
\begin{equation}
3[\text{Na}^+]_{\rm e}/g_{\rm Na}>2[\text{K}^+]_{\rm e}/g_{\rm K}, \label{32NaK}
\end{equation}
is satisfied, system \eqref{NaKCl}
possesses a steady state for sufficiently small $\alpha>0$.

To the best of the author's knowledge, there are
no analytical results on the stability of 
these steady states. In
\citep{weinstein1997dynamics} the author studies 
an epithelial model of greater complexity than \eqref{NaKCl}. 
The author obtains an algebraic expression for the linearized matrix
around steady state and numerically studies its eigenvalues for 
physiological parameter values. The computation of this linearization 
is complicated by the presence of an algebraic constraint
(the electroneutrality condition). Some authors 
have considered simpler {\em non}-electrolyte models of (epithelial) 
cell volume control and solute transport.
Analytical results for such models can be found in
\citep{weinstein1992analysis,hernandez2003stability,hernandez2007general,
benson2010general}. 

The goal of this paper is to establish analytical results on the 
existence and stability of steady states for a large class of 
pump-leak models that includes \eqref{NaKCl} and other representative 
models as special cases. In Section \ref{modelformulation}, we 
introduce the general class of pump-leak models that we shall treat 
in this paper. We consider $N$-species of ions subject to the 
electroneutrality constraint. Cell volume is controlled by 
the transmembrane osmotic pressure difference. 
The key observation of this Section is that the system of equations 
satisfies a free energy identity. This identity and its variations
will be the main tool in studying the stability of steady states.
The presence of a free energetic structure in pump-leak models leads 
us to natural structure conditions to be imposed on current-voltage 
relationships for ionic channel currents. 

In Section \ref{lin}, we study the case when the ionic channel current
(or passive ionic flux) has a linear current-voltage relationship,
the pump currents are constant and water flow is linearly proportional 
to the transmembrane osmotic pressure difference. System \eqref{NaKCl} is an 
example of such a system. We first establish the necessary and sufficient 
condition under which the system possesses a unique steady state. 
This condition, when applied to system \eqref{NaKCl}, reduces 
to \eqref{NaKClcond} (in fact, our conclusion is slightly stronger; 
we shall see that the restriction $z\leq -1$ 
is not needed in \eqref{NaKClcond}).
We then prove that this steady state 
is globally asymptotically stable. 
If a steady state does not exist, the cell volume $v$ tends to infinity
as time $t\to \infty$ for any initial condition. 
The main tool in proving these statements 
is a modified version of the free energy 
introduced in Section \ref{modelformulation}. 
This modified free energy $\wt{G}$ satisfies $d\wt{G}/dt=-J, J\geq 0$, 
thus defining a Lyapunov function.
Asymptotic stability of steady states follows 
by an examination of $\wt{G}$. To prove the global statements,
we also make use of the fact that $J$ is a Lyapunov function.
This is a consequence of the fact that, in a suitable set 
of variables, the system is a gradient flow of the convex function $\wt{G}$
with respect to a suitable metric.
We thus have a clear dichotomy;
if the system has a steady state, it is globally asymptotically stable
and if not, the cell bursts.
In Section \ref{epimod}, we discuss a simple epithelial model
in which the cell is in contact with a mucosal and serosal bath.
When the current voltage relationships for the ionic channels 
are all linear, the same stability results hold 
for this simple epithelial model. 

In many modeling studies using the pump-leak model, 
the current-voltage relation 
for the ionic channel current is not linear. 
The popular Goldman current voltage relation is one such example. 
The goal of Section \ref{general} is to establish a result on the 
existence and stability of steady states with minimal assumptions 
on the current-voltage relation. Indeed, all we assume here are 
properties required on thermodynamic grounds.
We first discuss solvability of the differential algebraic system. 
Solvability is not entirely trivial given
the algebraic constraint of electroneutrality. 
We then show that for a general class of pump-leak models,
there is a steady state for sufficiently small positive pump rates 
so long as a generalization of condition \eqref{32NaK} is satisfied.
We then show that this steady state is asymptotically stable. 
Our main tool here is again the modified or relative free energy. 
The main difficulty in establishing stability is 
that the current voltage relation cannot in general be written as 
a function of the chemical potential jump only (this difficulty 
is not present when the current voltage relation is linear). 
We shall see that this effect is small when the pump rate is sufficiently 
small, which allows us to establish asymptotic stability of the 
steady state. 

\section{Model Formulation and the Free Energy Identity}
\label{modelformulation}

Consider $N$ species of ion and let 
$c_k, k=1,\cdots, N$ be the intracellular
concentrations of the $k$-th species of ion. 
We let $c_k^{\rm e}$ be the extracellular concentrations 
of these ions, which are assumed to be positive and 
constant independent of time. 
Let $v$ be the volume of the cell. The balance equation 
for the ions can be written as follows:
\begin{equation}\label{ions}
\D{}{t}(v c_k)
=-j_k(\phi,\mb{c},\mb{c}^{\rm e})-p_k(\phi,\mb{c},\mb{c}^{\rm e}), \; 
k=1,\cdots,N.
\end{equation}
Here, we have written the transmembrane flux as the sum 
of the passive flux $j_k$ and the active flux $p_k$.
The active flux, typically generated by ionic pumps, 
requires energy expenditure whereas the passive flux,
carried by ionic channels and transporters, does not. 
The flux functions
$j_k$ and $p_k$ depend on the transmembrane potential $\phi$ 
(intracellular potential minus the extracellular potential)
as well as the vector of intracellular and extracellular
concentrations $\mb{c}=(c_1,\cdots,c_N)^T$ and 
$\mb{c}^{\rm e}=(c_1^{\rm e},\cdots,c_N^{\rm e})^T$, 
where $\cdot^T$ denotes the transpose.
We assume that $j_k$ and $p_k$ are $C^1$ functions of their arguments.
Since the $\mb{c}^{\rm e}$ are assumed constant, the $\mb{c}^{\rm e}$
only appear as parameters of the above differential equation. 
The dependence on $\mb{c}^{\rm e}$ will thus often not be shown explicitly.

The functional form of the active flux function $p_k$ can 
be arbitrary, but for $j_k$, we must 
impose structure conditions so that it represents a passive flux. 
Commonly used examples of $j_k$ are:
\begin{align}\label{jkL}
j_k^{\rm L}&=G_k\paren{RT\ln \paren{\frac{c_k}{c_k^{\rm e}}}+Fz_k\phi}
=G_k\mu_k,\; G_k>0,\\
\nonumber j_k^{\rm GHK}&=P_k\frac{Fz_k\phi}{RT}
\frac{c_k\exp\paren{\frac{Fz_k\phi}{RT}}-c_k^{\rm e}}
{\exp\paren{\frac{Fz_k\phi}{RT}}-1}\\
&=P_kc_k^{\rm e}\frac{Fz_k\phi/RT}
{\exp\paren{\frac{Fz_k\phi}{RT}}-1}
\paren{\exp\paren{\frac{\mu_k}{RT}}-1}\label{jkGHK}, \; P_k>0.
\end{align}
where $RT$ is the ideal gas constant times absolute temperature, 
$F$ is the Faraday constant and 
$z_k$ is the valence of the $k$-th species 
of ion (e.g., $1$ for Na$^+$, $-1$ for Cl$^-$ and so on).
We assume that there is at least one ionic species for which $z_k\neq 0$.
In the above, $\mu_k$ is the chemical potential of the $k$-th intracellular 
ion measured with respect to the extracellular bath:
\begin{equation}
\mu_k=RT\ln \paren{\frac{c_k}{c_k^{\rm e}}}+Fz_k\phi.\label{muk}
\end{equation}
Expression $j_k^{\rm L}$ is linear in $\mu_k$. If we multiply this 
by $Fz_k$ to change units from ion flux into electric current,
we obtain a linear current voltage relationship for this ionic current.
This was used in \eqref{NaKCl}. Expression $j_k^{\rm GHK}$ can be derived 
by assuming a constant electric field across the ion channel, 
and is known as the Goldman-Hodgkin-Katz current formula.
An important feature of both $j_k^{\rm L}$ and $j_k^{\rm GHK}$
is that they are increasing functions of $\mu_k$ for fixed $\phi$
and that it is $0$ when $\mu_k$ is $0$. 
In the general case, we require $j_k$ to satisfy a somewhat weaker version
of these properties, which we shall discuss later in relation 
to Proposition \ref{FEE}. 

Observe that $c_k$ can be expressed
in terms of $\mu_k$ and $\phi$ through \eqref{muk}. 
We shall often find it useful to view $j_k$ as a
function of $\phi$ and $\bm{\mu}=(\mu_1,\cdots,\mu_N)^T$ instead of 
$\phi$ and $\mb{c}$. We shall write this as
$j_k(\phi,\bm{\mu})$ in a slight abuse of notation.
We note that the dependence of $j_k$ on $\mu_l, l\neq k$
expresses the possibility that the flow of the $k$-th ion may 
be driven by the chemical potential gradient of the $l$-th ion.  
This is the case with many ionic transporters in which the 
flow of one species of ion is coupled to another. Indeed, 
many models of ionic transporter currents have this feature 
\citep{weinstein1983nonequilibrium,strieter_volume-activated_1990}. 

Equations \eqref{ions} are supplemented by:
\begin{align}
0&=\sum_{k=1}^N Fz_kc_k+\frac{FzA}{v}
=\sum_{k=1}^N Fz_kc_k^{\rm e},\label{EN}\\
\D{v}{t}&=-j_{\rm w}(\mb{c},\mb{c}^{\rm e},v).
\label{vol}
\end{align}
In \eqref{EN},
$A>0$ is the total amount of organic molecules inside the cell,
and $z$ is the average valence of intracellular organic molecules. 
Equation \eqref{EN} states that both the intracellular and 
extracellular concentrations satisfy the electroneutrality 
constraint. Electroneutrality of the extracellular space, together 
with the requirement that $z_k\neq 0$ for at least one $k$, requires that 
there must be at least two ionic species. 
In \eqref{vol}, $j_{\rm w}$ is the passive transmembrane water flux.
An example of $j_{\rm w}$ is:
\begin{equation}\label{jwexp}
j_{\rm w}=-\zeta\pi_{\rm w}, 
\; \pi_{\rm w}=RT\paren{\sum_{k=1}^N c_k^{\rm e}
-\paren{\sum_{k=1}^N c_k+\frac{A}{v}}}
\end{equation}
where $\zeta>0$ is the hydraulic conductivity of water through the membrane.
This simple prescription is what is 
used in \eqref{NaKCl} and many other models of 
cell volume and electrolyte control.
In this example, water flow is proportional to $\pi_{\rm w}$, the 
osmotic pressure difference across the membrane,  
whose expression is given by the van't Hoff law. 
We shall assume that $j_{\rm w}$ is a $C^1$ function only of $\pi_{\rm w}$.
The important property of \eqref{jwexp} is that 
$j_{\rm w}=0$ when $\pi_{\rm w}=0$ 
and that it is increasing in $\pi_{\rm w}$.
This property will be discussed further in relation 
to Proposition \ref{FEE}.

We seek solutions to the differential algebraic system 
\eqref{ions}, \eqref{EN} and \eqref{vol} given initial values for
$\mb{c}=(c_1,\cdots, c_N)^T, v$ and $\phi$ 
that satisfy the algebraic constraint \eqref{EN}.
We require that $c_k>0, k=1,\cdots, N$ and $v>0$ for all time. 
The membrane potential $\phi$ evolves so that the electroneutrality 
constraint \eqref{EN} is satisfied at each instant. Multiplying 
\eqref{ions} by $Fz_k$ and summing in $k$, we have:
\begin{equation}
I(\phi,\mb{c})=\sum_{k=1}^N Fz_k(j_k(\phi,\mb{c})+p_k(\phi,\mb{c}))=0,
\label{ENI}
\end{equation}
where we used \eqref{EN} to conclude that $I$, 
the total transmembrane electric current, must be $0$. 
This gives us an algebraic equation for 
$\phi$ that must be solved at each instant. The solvability of
this equation is a necessary condition for the initial value problem 
to be solvable. For certain specific functional 
forms of $j_k$ and $p_k$, like those of \eqref{jkL} and \eqref{jkGHK}
with $p_k$ constant,
the solvability of \eqref{ENI} is immediate. 
We shall discuss the general case in Section \ref{general}.

One way to avoid the above difficulty arising from the algebraic 
constraint 
is to replace \eqref{EN} with the following relation for 
intracellular concentrations:
\begin{equation}
C_{\rm m}\phi=F\paren{v\sum_{k=1}^N z_kc_k+zA},\label{cmphi}
\end{equation}
where $C_{\rm m}$ is the total membrane capacitance. This
states that the total charge inside the cell is equal to the 
charge stored on the membrane capacitor. 
We note that
\eqref{cmphi} has its own difficulties as a biophysical model. If the 
concentrations are defined as the total amount of intracellular 
ions divided by volume, \eqref{cmphi} is indeed correct. 
If we define $c_k$ to be the ionic concentration away from 
the surface charge layer (on the order of Debye length $\approx 1$nm 
in width), we must introduce surface ionic densities as was done in 
\citep{CAMCoS}.
As we shall see shortly, the left hand side of \eqref{cmphi} is 
often negligibly small and the electroneutrality condition \eqref{EN}
can thus be seen as a perturbative limit of condition \eqref{cmphi}. 
We shall discuss this point further after we make our system 
dimensionless. 

Scale ionic concentration, 
volume and membrane potential as follows 
and introduce the primed dimensionless variables:
\begin{equation}
c_k=c_0c_k', \; c_k^{\rm e}=c_0c_k^{\rm e\prime}, \; v=v_0v', \; \phi=\frac{RT}{F}\phi',
\end{equation} 
where $c_0$ and $v_0$ are the typical concentrations and volumes 
respectively. Equations \eqref{ions}, \eqref{EN} and \eqref{vol} become:
\begin{subequations}\label{system}
\begin{align}
\D{(v'c_k')}{t'}&=-j_k'(\phi',\bm{\mu}')-p_k'(\phi',\mb{c}'),\quad
\mu_k'=\ln \paren{\frac{c_k'}{c_k^{\rm e\prime}}}+z_k\phi', \label{ionsdless}\\
0&=\sum_{k=1}^N z_kc_k'+z\frac{A'}{v'}=\sum_{k=1}^N z_kc_k^{\rm e\prime}, \quad
A'=\frac{A}{c_0v_0},\label{ENdless}\\
\D{v'}{t'}&=-j_{\rm w}'(\pi_{\rm w}'),\quad  
\pi_{\rm w}'=\sum_{k=1}^N c_k^{\rm e\prime}-\paren{\sum_{k=1}^N c_k'+\frac{A'}{v'}},\label{voldless}
\end{align}
\end{subequations}
where $\bm{\mu}'=(\mu_1',\cdots,\mu_N')^T$ and $\mb{c}'=(c_1',\cdots,c_N')^T$.
Time $t$ and the flux functions $j_k, p_k$ and $j_{\rm w}$ 
are suitably rescaled to yield their respective primed variables.
We note that it is possible to further reduce the number of
constants, for example, by taking $c_0$ to be the total extracellular 
concentration. We shall not pursue this here, since it leads to some 
difficulty in understanding the physical meaning of each term in 
the resulting dimensionless system.

Equation \eqref{cmphi} yields:
\begin{equation}
\epsilon \phi'=\paren{v'\sum_{k=1}^N z_kc_k'+zA'},\quad 
\epsilon=\frac{C_{\rm m}RT/F}{Fc_0v_0}\label{cmphidless}
\end{equation}
where $\epsilon$ is a dimensionless parameter expressing the 
ratio between the amount of ions 
contributing to the surface charge and 
the absolute amount of charge in the cytosolic bulk. 
This quantity is typically very small (about $10^{-7}$) and 
we thus expect that it is an excellent approximation to let 
the left hand side of \eqref{cmphidless} be $0$ and adopt condition 
\eqref{EN} (or its dimensionless version \eqref{ENdless}), 
if the membrane potential does not vary too rapidly.
Most, if not all modeling studies of cellular electrolyte 
and water balance use the electroneutrality constraint \eqref{EN} or
\eqref{ENdless}
and we shall treat this case only. 

We shall henceforth deal almost exclusively with the dimensionless 
system. To avoid cluttered notation, we remove the primes from 
the dimensionless variables. 

An important property of the above system is that it possesses 
a natural energy function.
\begin{proposition}\label{FEE}
Let $v$ and $c_k, k=1,\cdots, N$ satisfy system \eqref{system}.
Then, the following equality holds:
\begin{equation}\label{FE}
\D{G}{t}=-\sum_{k=1}^N \mu_k(j_k+p_k)-\pi_{\rm w}j_{\rm w},
\end{equation}
where 
\begin{equation}\label{sigdef}
G=v\sigma, \quad 
\sigma=\sum_{k=1}^N \paren{c_k\paren{\ln \paren{\frac{c_k}{c_k^{\rm e}}}-1}+c_k^{\rm e}}
+ \frac{A}{v}\paren{\ln \paren{\frac{A}{v}}-1}.
\end{equation}
\end{proposition}
Identity \eqref{FE} is not entirely new. 
In \citep{sauer1973nonequilibrium,fromter1974electro,weinstein1983nonequilibrium}
the authors argue on thermodynamic grounds that 
free energy dissipation and input for an epithelial system should 
be given by the right hand side of \eqref{FE}. 
The new observation here is that, by appropriately defining a 
free energy function (that is to say, the left hand side of \eqref{FE}), 
this thermodynamic property can be turned into a mathematical
statement about system \eqref{system}. 
We also point out that a similar identity valid for a 
system of partial differential equations describing 
electrodiffusion and osmosis was proved in \citep{mori_liu_eis}.
In subsequent sections 
we will use this as a tool to study stability of steady states. 

\begin{proof}[Proof of Proposition \ref{FEE}]
View $\sigma$ as a function of $c_k$ and $c_A=A/v$. 
Note that:
\begin{equation}
\D{}{t}(vc_A)=\D{A}{t}=0,\label{vat}
\end{equation}
since $A$ is constant. Define the chemical potential of 
intracellular organic molecules as:
\begin{equation}
\mu_A=\ln c_A+z\phi=\PD{\sigma}{c_A}+z\phi.
\end{equation}
Note that:
\begin{equation}
\mu_k=\PD{\sigma}{c_k}+z_k\phi.
\end{equation}
Multiply \eqref{ions} by $\mu_k$, 
multiply \eqref{vat} by $\mu_A$ and take the sum. The left hand 
side yields:
\begin{equation}
\begin{split}
&\sum_{k=1}^N\mu_k\D{}{t}(vc_k)+\mu_A\D{}{t}(vc_A)\\
=&\sum_{k=1}^N\paren{\PD{\sigma}{c_k}+z_k\phi}\D{}{t}(vc_k)
+\paren{\PD{\sigma}{c_A}+z\phi}\D{}{t}(vc_A)\\
=&\D{}{t}(v\sigma)+
\paren{\sum_{k=1}^Nc_k\PD{\sigma}{c_k}+c_A\PD{\sigma}{c_A}-\sigma}\D{v}{t}
+\phi \D{}{t}\paren{v\sum_{k=1}^N z_kc_k+zvc_A}\\
=&\D{}{t}(v\sigma)-\pi_{\rm w}\D{v}{t},
\end{split}
\end{equation}
where we used \eqref{EN} and \eqref{sigdef} 
in the third equality. 
We thus have:
\begin{equation}
\D{}{t}(v\sigma)-\pi_{\rm w}\D{v}{t}
=-\sum_{k=1}^N \mu_k(j_k+p_k).
\end{equation}
Equation \eqref{FE} thus follows from \eqref{vol}.
\end{proof}
The function $\sigma$ should be interpreted as the free energy 
per unit volume of intracellular electrolyte solution. 
A key fact that was used in the above proof is the identity:
\begin{equation}
\pi_{\rm w}=\sigma-\paren{\sum_{k=1}^Nc_k\PD{\sigma}{c_k}+c_A\PD{\sigma}{c_A}}.
\end{equation}
This relation, connecting the free energy with osmotic pressure, 
is well-known in physical chemistry \citep{doi1996introduction}.

When $p_k=0, \; k=1,\cdots,N$ in \eqref{FE}, there are no active currents and we have:
\begin{equation}
\D{G}{t}=-\sum_{k=1}^N \mu_kj_k-\pi_{\rm w}j_{\rm w}.
\label{vsigpassive}
\end{equation}
Given the interpretation of $G$ as the total 
free energy of the system, the second law of thermodynamics requires 
that $G$ be decreasing in time. 
The negativity of \eqref{vsigpassive} when $\bm{\mu}\neq \mb{0}$
and $\pi_{\rm w}\neq 0$ is equivalent to the statement that 
the following conditions be satisfied:
\begin{align}
\sum_{k=1}^N \mu_kj_k(\phi,\bm{\mu}) &> 0 \text{ for all } \phi 
\text{ and } \bm{\mu}\neq \mb{0},
\label{jkgencond}\\
\pi_{\rm w}j_{\rm w}(\pi_{\rm w})&>0 \text{ if } \pi_{\rm w}\neq 0,
\label{jwgencond}
\end{align}
Condition \eqref{jkgencond}, together with 
continuity of $j_k$ and $j_{\rm w}$ with respect to its arguments 
immediately implies that:
\begin{equation}
j_k(\phi,\bm{\mu}=\mb{0})=0,\; k=1,\cdots,N,\quad  
\; j_{\rm w}(\pi_{\rm w}=0)=0,\label{noflux}
\end{equation}
where
the conditions on $j_k$ is to be satisfied for all $\phi$.
Taking the derivative of the above expression for $j_k$
with respect to $\phi$, we see that: 
\begin{equation}
\PD{j_k}{\phi}(\phi,\bm{\mu}=\mb{0})=0, \; k=1,\cdots,N.\label{djkdphi}
\end{equation}
We shall find this expression useful later on.

Let us require 
that the derivative of $j_{\rm w}$ with respect to $\pi_{\rm w}$ be non-zero
at $\pi_{\rm w}=0$. This non-degeneracy condition, together with 
\eqref{jwgencond}, leads to:
\begin{equation}\label{jwinc}
\PD{j_{\rm w}}{\pi_{\rm w}}(\pi_{\rm w}=0)>0.
\end{equation}
Let $\mb{j}=(j_1,\cdots,j_N)^T$ and let $\partial\mb{j}/\partial\bm{\mu}$ be
the Jacobian matrix with respect to $\mb{\mu}$
for fixed $\phi$. That is to say, the $kl$ entry of the $N\times N$
matrix $\partial \mb{j}/\partial \bm{\mu}$ is given by 
$\partial j_k/\partial\mu_l$.
The non-degeneracy condition 
for $\mb{j}$ is that $\partial\mb{j}/\partial\bm{\mu}$ be non-singular
at $\bm{\mu}=0$. 
We require the following condition:
\begin{equation}
\PD{\mb{j}}{\bm{\mu}}(\phi,\bm{\mu}=\mb{0}) 
\text{ is symmetric positive definite for all } \phi.
\label{posdef}
\end{equation}
The symmetry of the Jacobian matrix does {\em not} follow from the
non-degeneracy condition and condition \eqref{jkgencond}. 
In fact, these two conditions imply only that:
\begin{equation}\label{weakposdef}
\frac{1}{2}\paren{
\PD{\mb{j}}{\bm{\mu}}+\paren{\PD{\mb{j}}{\bm{\mu}}}^T}(\phi,\bm{\mu}=\mb{0}) 
\text{ is symmetric positive definite for all } \phi.
\end{equation}
The symmetry of 
$\partial\mb{j}/\partial\bm{\mu}$ is required 
by the Onsager reciprocity principle 
\citep{onsager1931reciprocal,katzir1965nonequilibrium,kjelstrup2008non}.
We note that \eqref{posdef} is the same as the condition introduced in 
\citep{sauer1973nonequilibrium,fromter1974electro,
weinstein1983nonequilibrium}.
It is easy to see that both \eqref{jkL} and \eqref{jkGHK} 
satisfy \eqref{jkgencond}, \eqref{noflux} and \eqref{posdef}
and that \eqref{jwexp} satisfies \eqref{jwgencond}, \eqref{noflux}
and \eqref{jwinc}. 

Note that \eqref{jwgencond} together with \eqref{voldless} implies that 
the intracellular and extracellular osmotic pressures
must be equal at steady state. 
We are thus assuming that the membrane 
cannot generate any mechanical force to balance 
a difference in osmotic pressure. 
If the cell membrane (or its attendant structures) 
can generate some elastic force, 
it would make it easier 
for the cell to maintain its volume. 

Our starting point in deriving the above structure conditions 
for $j_k$ and $j_{\rm w}$ was the requirement that the right hand 
side of \eqref{vsigpassive} be negative. If $j_k$ is allowed 
to depend on $\pi_{\rm w}$ and $j_{\rm w}$ on $\bm{\mu}$, a more general 
structure condition can be formulated.  
Although it should not be difficult
to extend the results to follow to this more general case,
we will not pursue this here to keep the presentation reasonably simple. 

\section{Results when the 
Flux Functions are Linear in the Chemical Potential Jump}\label{lin}

Before dealing with the general case in Section \ref{general}, we 
treat the simpler case when 
$j_k$ and $j_{\rm w}$ are linear in $\bm{\mu}$ and $\pi_{\rm w}$
respectively and $p_k$ are constants in \eqref{system}.
In this case, we obtain a more or less
complete picture of the behavior of our system.
System \eqref{system} becomes:
\begin{subequations}\label{linsys}
\begin{align}
\D{}{t}(v\mb{c})&=-L\bm{\mu}-\mb{p},\label{linsys1}\\
0&=\sum_{k=1}^N z_kc_k +\frac{zA}{v}=\sum_{k=1}^N z_kc_k^{\rm e}\label{linsysEN},\\
\D{v}{t}&=-\zeta \pi_{\rm w}\label{linsys2},
\end{align}
\end{subequations}
where $L=\partial \mb{j}/\partial \bm{\mu}$
is an $N\times N$ matrix, which by \eqref{posdef}, is symmetric 
positive definite, and $\mb{p}=(p_1,\cdots,p_N)^T$ is 
the vector of the active fluxes. 
The hydraulic permeability $\zeta=\partial j_{\rm w}/\partial \pi_{\rm w}$ 
is positive by \eqref{jwinc}.
The extracellular ionic concentrations 
$c_k^{\rm e}>0, k=1,\cdots,N$ and 
the amount of impermeable organic solute $A$ are assumed positive
as discussed in the previous Section.
We seek solutions $(\mb{c},v)\in \mathbb{R}_+^{N}\times \mathbb{R}_+$
where $\mathbb{R}_+$ denotes the set of positive real numbers.

Here and in the sequel, we shall often find it useful to 
refer to the pair $(\mb{c},v)$ as well as the triple 
$(\mb{c},v,\phi)\in \mathbb{R}_+^{N}\times \mathbb{R}_+\times \mathbb{R}$.
We shall often view the pair and the triple 
as being members of $\mathbb{R}_+^{N+1}$
and $\mathbb{R}_+^{N+1}\times \mathbb{R}$ respectively and write
$(\mb{c},v)\in \mathbb{R}_+^{N+1}$
and $(\mb{c},v,\phi)\in \mathbb{R}_+^{N+1}\times \mathbb{R}
\subset\mathbb{R}^{N+2}$ .
The more ``correct'' notation may be to write 
$(\mb{c}^T,v)^T=(c_1,\cdots,c_N,v)^T\in \mathbb{R}_+^{N+1},\;
(\mb{c}^T,v,\phi)^T
=(c_1,\cdots,c_N,v,\phi)^T\in \mathbb{R}_+^{N+1}\times \mathbb{R}$
given that $\mb{c}$ is a column vector. 
We will not adopt this unnecessarily ugly notation. 
Similar comments apply to the pair 
$(\mb{a},v)$ and the triple $(\mb{a},v,\phi)$ where $\mb{a}=v\mb{c}$.

For system \eqref{linsys}, we may compute $\phi$ explicitly in terms 
of $\mb{c}$ by solving the (dimensionless version of) \eqref{ENI}:
\begin{equation}\label{Lzgamma}
\phi=-\frac{\ip{\mb{z}}{L\bm{\gamma}+\mb{p}}_{\mathbb{R}^N}}
{\ip{\mb{z}}{L\mb{z}}_{\mathbb{R}^N}},\quad
\bm{\gamma}=(\gamma_1,\cdots,\gamma_N)^T, \; 
\gamma_k=\ln\paren{\frac{c_k}{c_k^{\rm e}}},
\end{equation}
where $\mb{z}=(z_1,\cdots,z_N)^T$.
Substituting the above expression for $\phi$ into \eqref{linsys1},
we obtain:
\begin{equation}\label{whLqdef}
\begin{split}
\D{}{t}(v\mb{c})&=-\wh{L}(\bm{\gamma}+\mb{q}), \quad \wh{L}_{kl}=L_{kl}
-\frac{(L\mb{z})_k(L\mb{z})_l}{\ip{\mb{z}}{L\mb{z}}_{\mathbb{R}^N}},\\
\mb{q}&=(q_1,\cdots,q_N)^T,\quad \mb{q}=L^{-1}\mb{p},
\end{split}
\end{equation}
where $L_{kl}, \wh{L}_{kl}$ are the $kl$ entries
 of the $N\times N$ matrix $L, \wh{L}$ and $(L\mb{z})_k$ is the 
$k$-th component of the vector $L\mb{z}\in \mathbb{R}^N$.
Note that $L^{-1}$ exists given that $L$ is positive definite.
Replacing \eqref{linsys1} with \eqref{whLqdef}, 
we now have a system of ordinary differential equations (ODEs)
for $\mb{c}$ and $v$ only.
Applying the standard existence and uniqueness theorem for 
ODEs, we conclude that a unique solution always exist for 
sufficiently short time so long as the solution 
remains in $(\mb{c},v)\in \mathbb{R}_+^{N+1}$.

Even though our ODE system 
is for $N+1$ variables $\mb{c}=(c_1,\cdots,c_N)^T$ 
and $v$, our dynamical system is only $N$
dimensional, since the dynamics is constrained by  
the electroneutrality condition \eqref{linsysEN}.
The initial value problem for \eqref{linsys} can thus only 
be solved if the initial values satisfy \eqref{linsysEN}.
In the proof of Proposition \ref{linmain}, we will find it 
useful to make a change of variables to remove this constraint.

\subsection{Existence of Steady States and Asymptotic Stability}
\label{linsteadglobal}

Our first observation is the following.
\begin{proposition}\label{linextcond}
Consider the function:
\begin{equation}\label{fphidef}
f(\phi)=\sum_{k=1}^N c_k^{\rm e}\paren{\exp(-q_k-z_k\phi)-1}, \; 
\end{equation}
where $q_k, k=1,\cdots,N$ was defined in \eqref{whLqdef}.
The function $f(\phi), \phi\in \mathbb{R}$ has a unique 
minimizer $\phi=\phi_{\rm min}$.
System \eqref{linsys} has a unique steady state if this 
minimum value is negative:
\begin{equation}\label{fphimin}
f_{\rm min}(\mb{q},\mb{c}^{\rm e}, \mb{z})\equiv f(\phi_{\rm min})<0.
\end{equation}
Otherwise, the system does not have any steady states.
\end{proposition}
The above condition can be interpreted as follows. At 
steady state, it is easily seen that the concentrations $c_k$ must be equal to 
$c_k^*=c_k^{\rm e}\exp(-q_k-z_k\phi^*)$ where $\phi^*$ is the 
value of $\phi$ at steady state (see \eqref{ckststate}).
We need $f(\phi^*)=\sum_{k=1}^N (c_k^*-c_k^{\rm e})<0$ since 
there must be ``osmotic room'' for the impermeable solutes. 
This is only possible if the minimum of 
$f(\phi), \phi\in \mathbb{R}$ is negative. 
We also point out that the above condition depends only 
on $\mb{q}=L^{-1}\mb{p},\mb{c}^{\rm e}$ and $\mb{z}$ and does not depend 
on $z$ or $A$. 
\begin{proof}[Proof of Proposition \ref{linextcond}]
Set the right hand side of \eqref{linsys1} and \eqref{linsys2} to zero. 
We have:
\begin{equation}
\mb{q}=L^{-1}\mb{p}=-\bm{\mu}, \quad \pi_{\rm w}=0.
\end{equation}
Solving for $c_k$ in the first expression we have:
\begin{equation}\label{ckststate}
c_k=c_k^{\rm e}\exp(-q_k-z_k\phi),\quad k=1\cdots,N.
\end{equation}
Substitute this into $\pi_{\rm w}=0$ and \eqref{linsysEN}. We have:
\begin{align}
f(\phi)+\frac{A}{v}&=0,
\label{fphi}\\
-\D{f}{\phi}+\frac{zA}{v}&=
\sum_{k=1}^Nz_kc_k^{\rm e}\exp(-q_k -z_k\phi)
+\frac{zA}{v}=0,\label{Qphi}
\end{align}
where $f(\phi)$ is given by \eqref{fphidef}.
We must find solutions $\phi$ and $v>0$ to the above system.
Note that:
\begin{equation}\label{Pphisign}
\DD{2}{f}{\phi}=\sum_{k=1}^Nz_k^2c_k^{\rm e}\exp(-q_k-z_k\phi)>0, \quad
\lim_{\phi\pm \infty} \D{f}{\phi}=\pm \infty.
\end{equation}
The second property comes from the fact that there are ions with 
negative and positive valences among the $N$ species of ions 
and that $c_k^{\rm e}>0$. We can thus solve \eqref{Qphi} for 
$\phi$ uniquely in terms of $v$. Let this function be
$\phi=\varphi(v)$. We have:
\begin{equation}\label{phiv}
\D{\varphi}{v}=-\paren{\DD{2}{f}{\phi}}^{-1}\frac{zA}{v^2}.
\end{equation} 
Consider the left hand side of \eqref{fphi} and substitute $\phi=\varphi(v)$
into this expression:
\begin{equation}\label{Rveq}
R(v)\equiv f(\varphi(v))+\frac{A}{v}=0.
\end{equation}
Our problem of finding steady states is reduced to the question of 
whether the above equation in $v$ has a positive solution.
We have:
\begin{equation}
\D{R}{v}=\D{f}{\phi}\D{\varphi}{v}-\frac{A}{v^2}
=-\paren{\DD{2}{f}{\phi}}^{-1}\frac{(zA)^2}{v^3}-\frac{A}{v^2}
\leq -\frac{A}{v^2}<0
\end{equation}
where we used \eqref{Qphi} and \eqref{phiv} in the second equality
and \eqref{Pphisign} in the first inequality. 
Therefore, $R(v)$ is monotone decreasing. Note that:
\begin{equation}
R(\epsilon)
=R(1)-\int_{\epsilon}^1 \paren{\D{R}{v}}dv
\geq R(1)+\int_{\epsilon}^1 \paren{\frac{A}{v^2}}dv=R(1)
+A(\epsilon^{-1}-1).
\end{equation}
Therefore, $R(v)\to \infty$ as $v$ tends to $0$ from above. 
Thus, \eqref{Rveq} has a unique positive solution if:
\begin{equation}\label{Rvneg}
\lim_{v\to \infty} R(v)<0,
\end{equation} 
and otherwise, there is no solution.
Let $\varphi_\infty=\lim_{v \to \infty}\varphi(v)$.
Note that this limit exists since, by \eqref{phiv}, $\varphi(v)$
is monotone if $z\neq 0$ and constant if $z=0$. 
Taking the limit $v\to \infty$ on both sides of \eqref{Qphi}, 
we see that $\varphi_\infty$ is the unique solution to $df/d\phi=0$
as an equation for $\phi$. Given \eqref{Rveq},
condition \eqref{Rvneg} can be written as $f(\varphi_\infty)<0$.
The statement follows by taking $\phi_{\rm min}=\varphi_\infty$.
\end{proof}
Condition \eqref{fphimin}
applied to \eqref{NaKCl} yields condition 
\eqref{NaKClcond}. Since condition \eqref{fphimin} is valid regardless 
of the value of $z$, we may 
lift the restriction $z\leq -1$ found in \citep{KS}.

Fix $\mb{q}=L^{-1}\mb{p}$ and $\mb{c}^{\rm e}$
so that \eqref{fphimin} is satisfied. Since condition \eqref{fphimin}
does not depend on $z$ or $A$, 
a unique steady state $(\mb{c},v,\phi)=(\mb{c}^*,v^*,\phi^*)$ 
exists for any value of $z$ and $A>0$. We may thus view
$(\mb{c}^*,v^*,\phi^*)$
as a function of $A$ and $Q_A=zA$, defined for $A>0$ and $Q_A\in \mathbb{R}$.
We can compute the dependence of $v^*$ on $Q_A$ and $A$
as follows. 
\begin{align}
\PD{v^*}{A}&=\paren{\paren{\frac{Q_A}{v^*}}^2
+\DD{2}{f}{\phi}\frac{A}{v^*}}^{-1}
\DD{2}{f}{\phi}>0,\label{vA}\\
\PD{v^*}{Q_A}&=\paren{\paren{\frac{Q_A}{v^*}}^2+\DD{2}{f}{\phi}\frac{A}{v^*}}^{-1}\frac{Q_A}{v^*}.\label{vQA}
\end{align}
From a biophysical standpoint, 
\eqref{vA} is reasonable since more
impermeable solute leads to greater 
osmotic pressure. Note that \eqref{vQA} says that the $v^*$ increases 
if the absolute amount of charge (whether negative or positive) increases. 
This is also biophysically reasonable since more charge 
on the impermeable solute leads to a greater amount of intracellular 
counterions, thus increasing intracellular osmotic pressure. 

We now turn to the question of stability. Let:
\begin{equation}\label{icS}
\mathcal{S}=\lbrace (\mb{c},v)\in \mathbb{R}_+^{N+1}|
\sum_{k=1}^N z_kc_k +zA/v=0\rbrace.
\end{equation} 
The dynamical system defined by \eqref{linsys} lives on this set.
We must thus modify the definition of stability accordingly.
A steady state of \eqref{linsys} is {\em stable} 
if all solutions with initial values in $\mathcal{S}$ 
and near the steady state stay close to the 
steady state. 
A steady state is {\em asymptotically stable} 
if it is stable and if all solutions with initial values
in $\mathcal{S}$ and 
near the steady state
converge to the steady state as $t\to \infty$. 
A steady state is {\em globally asymptotically stable} 
if it is stable and if all solutions starting from initial 
values in $\mathcal{S}$ converge to 
the steady state as $t\to \infty$.
If we make a change of variables 
to obtain an $N$-dimensional dynamical system without 
the implicit constraint of electroneutrality,
the above definitions of stability reduce to the usual ones
for ODEs.

We saw in Proposition \ref{FEE} that
in the absence of active currents, the free energy $G$ defined 
in \eqref{sigdef} is decreasing. 
We now construct a free energy like quantity that is 
decreasing in the presence of active currents. 
Define:
\begin{equation}\label{wtG}
\wt{G}
=v\sum_{k=1}^N 
\paren{c_k\paren{\ln \paren{\frac{c_k}{c_k^{\rm e}}}-1+q_k}+c_k^{\rm e}}
+ A\paren{\ln \paren{\frac{A}{v}}-1}
\end{equation}
where $q_k$ was defined in \eqref{fphidef}.
For $\wt{G}$, we have the following analogue of Proposition \ref{FEE}.
\begin{lemma}\label{FEq}
Let $\mb{c},v, \phi$ satisfy system \eqref{linsys}. We have:
\begin{equation}
\begin{split}\label{GJdef}
\D{\wt{G}}{t}&=-J,\\
J&=\ip{\bm{\mu}+\mb{q}}
{L(\bm{\mu}+\mb{q})}_{\mathbb{R}^N}
+\zeta \pi_{\rm w}^2=
\ip{\bm{\gamma}+\mb{q}}{\wh{L}(\bm{\gamma}+\mb{q})}_{\mathbb{R}^N}
+\zeta \pi_{\rm w}^2,
\end{split}
\end{equation}
where $\ip{\cdot}{\cdot}_{\mathbb{R}^N}$ 
is the inner product in $\mathbb{R}^N$. The function $\wt{G}$
is thus a Lyapunov function in the sense that it is non-increasing 
in time.
\end{lemma}
\begin{proof}
The proof is almost identical to the proof of Proposition \ref{FEE}.
The second equality in the definition of $J$ comes from \eqref{whLqdef}.
\end{proof}
If system \eqref{linsys} has a steady state, we may rewrite 
\eqref{GJdef} as follows. 
Let $(\mb{c}^*,v^*,\phi^*)$ be the 
steady state of \eqref{linsys}. 
Note that $\pi_{\rm w}=0$ at steady state, and 
that $q_k=-\mu_k^*$ where $\mu_k^*$ is the evaluation of 
the chemical potential at steady state. Using this, and 
the fact that $\mb{c}$ satisfies the electroneutrality constraint 
\eqref{linsysEN}, we find, after some calculation:
\begin{equation}\label{hatG}
\begin{split}
&\wh{G}(\mb{c},v)\equiv \wt{G}(\mb{c},v)-\wt{G}(\mb{c}^*,v^*)\\
=&v\sum_{k=1}^N 
\paren{c_k\paren{\ln \paren{\frac{c_k}{c_k^*}}-1}+c_k^*}
+ A\paren{\ln \paren{\frac{v^*}{v}}-1+\frac{v}{v^*}}.
\end{split}
\end{equation}
The quantity $\wh{G}$ can be interpreted as being the total free 
energy of the system relative to the steady state.
Since $\wt{G}$ and $\wh{G}$ differ only by a constant, 
we may replace $\wt{G}$ with $\wh{G}$ in \eqref{GJdef}:
\begin{equation}\label{Ghateq}
\begin{split}
\D{\wh{G}}{t}&=-J=-\ip{\wh{\bm{\mu}}}{L\wh{\bm{\mu}}}_{\mathbb{R}^N}-\zeta \pi_{\rm w}^2,\\
\wh{\bm{\mu}}&=(\wh{\mu}_1,\cdots,\wh{\mu}_N)^T,\; 
\wh{\mu}_k=\ln \paren{\frac{c_k}{c_k^*}}+z_k(\phi-\phi^*),
\end{split}
\end{equation}
where we used $q_k=-\mu_k^*$ to rewrite $J$.
Thus, if the system has a steady state, Lemma \ref{FEq} says 
that the free energy relative to the steady state is always 
decreasing at a rate that is controlled by $\wh{\bm{\mu}}$,
the vector of chemical potential relative to the steady state. 
We shall find both $\wt{G}$ and $\wh{G}$ useful 
depending on context.

In studying stability, it is sometimes convenient to use a new set of 
variables $\mb{a}=(a_1,\cdots,a_N)^T=v\mb{c}, v$ and $\phi$ rather 
than $\mb{c}, v$ and $\phi$. 
Rewriting \eqref{linsys} in 
the new variables, $\mb{a}$ satisfies the differential equations:
\begin{equation}
\D{\mb{a}}{t}=-L(\bm{\mu}+\mb{q}).\label{Ceq}
\end{equation}

We now state a result on the function $\wt{G}$.
Note that, although the solutions $(\mb{c},v)$ of \eqref{linsys}
is defined only for $(\mb{c},v)\in \mathbb{R}_+^{N+1}$, 
the function $\wt{G}(\mb{c},v)$ is well-defined 
on $\overline{\mathbb{R}_+}^N\times \mathbb{R}_+$ where 
$\overline{\mathbb{R}_+}$ is the set of non-negative real numbers
(overline of a set will henceforth denote its closure).
The same comment applies for $\wt{G}$ viewed as a function of 
$(\mb{a},v)$ and for $\wh{G}$.

\begin{lemma}\label{propG}
\begin{enumerate}
\item Consider the function $\wt{G}$ defined in \eqref{wtG} and 
view this as a function of $(\mb{a},v)\in \overline{\mathbb{R}_+}^N\times \mathbb{R}_+$
where $\mb{a}=v\mb{c}$.
The Hessian of $\wt{G}$ is 
positive definite at each point in $\mathbb{R}_+^{N+1}$ and 
$\wt{G}$ is thus a convex function.
\item Suppose condition \eqref{fphimin} is satisfied. 
View $\wt{G}$ as a function of 
$(\mb{c},v)\in \overline{\mathbb{R}_+}^N\times \mathbb{R}_+$
and let the unique steady 
state of \eqref{linsys} be given by $(\mb{c}^*,v^*)$.
Then, $(\mb{c}^*,v^*)$ is the unique minimizer of $\wt{G}$
restricted to $\overline{\mathcal{S}}$, where $\mathcal{S}$ is 
given in \eqref{icS}.
\end{enumerate}
\end{lemma}
\begin{proof}
Let $H_G$ denote the $(N+1)\times (N+1)$
Hessian matrix of $\wt{G}(\mb{a},v)$. The Hessian matrix is well-defined 
for $(\mb{a},v)\in \mathbb{R}_+^{N+1}$.
For any vector $\mb{x}=(x_1,\cdots,x_N,x_v)\in \mathbb{R}^{N+1}$ 
we have:
\begin{equation}
\ip{\mb{x}}{H_G\mb{x}}_{\mathbb{R}^{N+1}}=\sum_{k=1}^N 
\paren{\frac{1}{\sqrt{a_k}}x_k-\frac{\sqrt{a_k}}{v}x_v}^2+\frac{A}{v^2}x_v^2,
\end{equation}
The Hessian matrix $H_G$ is thus
positive definite at every point in $(\mb{a},v)\in \mathbb{R}_+^{N+1}$, 
and $\wt{G}$ is thus a convex function.

To prove the second item, we first rephrase the assertion 
in terms of $(\mb{a},v)$.
We must show that $(\mb{a}^*,v^*)=(v^*\mb{c}^*,v^*)$
is the unique minimizer of 
$\wt{G}(\mb{a},v)$, 
$(\mb{a},v)\in \overline{\mathbb{R}_+}^N\times \mathbb{R}_+$ 
when restricted to the hyperplane:
\begin{equation}\label{QCk}
\sum_{k=1}^N z_ka_k+A=0 
\end{equation}
where $\mb{a}=(a_1,\cdots,a_N)^T$.

We seek stationary points of $\wt{G}$ restricted to the hyperplane \eqref{QCk}.
Consider:
\begin{equation}
\wt{G}_\lambda=\wt{G}+\lambda\paren{\sum_{k=1}^N z_ka_k +A}
\end{equation}
where $\lambda$ is the Lagrange multiplier. The condition 
for a stationary point is given by:
\begin{equation}
\begin{split}
\PD{\wt{G}_\lambda}{a_k}&=\ln\paren{\frac{a_k/v}{c^{\rm e}_k}}+z_k\lambda
+q_k=0,\\
\PD{\wt{G}_\lambda}{v}&=\pi_{\rm w}=0.
\end{split}
\end{equation}
If we identify $\lambda$ with $\phi$, the membrane potential, 
the above condition is nothing other than the condition for steady 
state of system \eqref{linsys}. Since condition \eqref{fphimin} is satisfied, 
by Proposition \ref{linextcond}, the above system has a unique solution and 
$(\mb{a},v,\lambda)=(\mb{a}^*,v^*,\phi^*)$ where $\phi^*$
is the value of $\phi$ at the unique steady state of \eqref{linsys}.
Since $\wt{G}$ is a convex function on $\overline{\mathbb{R}_+}^N\times \mathbb{R}_+$, 
its restriction to the hyperplane \eqref{QCk} is also a convex function.
Thus, this stationary point is the unique minimizer.
\end{proof}

We may now state our first stability result.
\begin{proposition}\label{linmain}
Suppose condition \eqref{fphimin} is satisfied. Then, 
the unique steady state of \eqref{linsys}
is asymptotically stable. Moreover, the decay to 
the steady state is exponential, and the  
linearized operator around the steady state is diagonalizable 
with real and negative eigenvalues.
\end{proposition}
The linearized operator above refers to the linearization when 
\eqref{linsys} is seen as an ODE system on the $N$-dimensional 
submanifold defined by the electroneutrality constraint \eqref{linsysEN}.
We note that asymptotic stability is in fact 
an immediate consequence of Lemma \ref{propG}
by a Lyapunov stability argument. Thus, if we are only interested 
in asymptotic stability, there is no need to study the linearization.
A Lyapunov stability argument will be used to study the global behavior of 
solutions in the proof of Theorem \ref{globalstab}.
\begin{proof}[Proof of Proposition \ref{linmain}]
It suffices to prove this claim by studying the 
dynamics of \eqref{linsys} in the variables $(\mb{a},v)$
instead of $(\mb{c},v)$.
Since condition \eqref{fphimin} is satisfied, 
there is a unique steady state by  
Proposition \ref{linextcond}.
To study the linearization around this steady 
state we must change variables to remove the implicit constraint 
\eqref{linsysEN} (or equivalently, \eqref{QCk}) 
and obtain an $N$-dimensional ODE system.

From \eqref{whLqdef}, we have:
\begin{equation}\label{eqnina}
\D{\mb{a}}{t}=-\wh{L}\nabla_a \wt{G}, 
\end{equation}
where $\nabla_a\wt{G}$ is the gradient of $\wt{G}(\mb{a},v)$ 
with respect to $\mb{a}$ while keeping $v$ fixed.
Let us examine the matrix $\wh{L}$. Take a vector $\mb{w}\in \mathbb{R}^N$.
By \eqref{whLqdef}, we have:
\begin{equation}\label{whatLw}
\ip{\mb{w}}{\wh{L}\mb{w}}_{\mathbb{R}^N}=
\ip{\mb{w}}{L\mb{w}}_{\mathbb{R}^N}-\frac{
\paren{\ip{\mb{w}}{L\mb{z}}_{\mathbb{R}^N}}^2}{\ip{\mb{z}}{L\mb{z}}_{\mathbb{R}^N}}.
\end{equation}
Since $L$ is symmetric positive definite, 
the above quantity is non-negative by the Cauchy-Schwarz inequality and is 
equal to $0$ if and only if $\mb{w}$ is a constant multiple of $\mb{z}$. 
Thus, $\wh{L}$ is a symmetric positive semi-definite matrix whose 
eigenspace corresponding to the eigenvalue $0$ is spanned by $\mb{z}$. 
The restriction of $\wh{L}$ to 
the orthogonal complement of this eigenspace is thus positive definite. 

Consider a change of coordinates from $\mb{a}$ to 
an orthonormal coordinate system $\mb{b}$ satisfying:
\begin{equation}
\mb{b}=(b_1,\cdots,b_N)^T, \; b_N=\frac{1}{\abs{\mb{z}}}\ip{\mb{z}}{\mb{a}}_{\mathbb{R}^N}.
\end{equation}
The $N$-th coordinate axis is thus
parallel to $\mb{z}$. Note that the electroneutrality constraint \eqref{QCk}
can be written as $b_N=\text{constant}$.
Let $\mb{b}=U\mb{a}$, where 
$U$ is the orthogonal coordinate transformation matrix. 
In the $\mb{b}$ coordinate system, 
$\wh{L}$ transforms to $\wh{L}_b=U\wh{L}U^{-1}$. 
Given the above properties of $\wh{L}$, 
$\wh{L}_b$ has the form:
\begin{equation}\label{whLU}
\wh{L}_b=\begin{pmatrix}
\wh{L}_b^\perp & \mb{0}_{N-1}\\
\mb{0}_{N-1}^T & 0
\end{pmatrix},
\end{equation}
where $\wh{L}_b^{\perp}$ is a positive definite matrix and $\mb{0}_{N-1}$
is the zero column vector of length $N-1$. 
Rewriting \eqref{eqnina} in terms of $\mb{b}$, we have: 
\begin{equation}
\D{\mb{b}}{t}=-\wh{L}_b\nabla_{b} \wt{G}
\end{equation}
where $\nabla_b \wt{G}$ is the gradient of $\wt{G}$ (keeping $v$ fixed) 
seen as a function of $\mb{b}$. From \eqref{whLU}, we see that 
$b_N$ remains constant. Let $\wh{\mb{b}}=(b_1,\cdots,b_{N-1})$.
We have:
\begin{equation}\label{LUG}
\D{\wh{\mb{b}}}{t}=-\wh{L}_b^\perp\nabla_{\wh{b}}\wt{G}
\end{equation}
where $\nabla_{\wh{b}}$ is the gradient of $\wt{G}$ with respect to 
$\wh{\mb{b}}$ while keeping $b_N$ and $v$ fixed. 
We have thus reduced the system \eqref{linsys} 
to \eqref{LUG} and to \eqref{linsys2}
which can be written as:
\begin{equation}
\D{v}{t}=-\zeta \PD{\wt{G}}{v}.
\end{equation}
Letting $\mb{u}=(b_1,\cdots,b_{N-1},v)^T$, 
we may write our system as follows:
\begin{equation}\label{ELE}
\D{\mb{u}}{t}=-L_u\nabla_u \wt{G}, \quad
L_u=
\begin{pmatrix}
\wh{L}_b^\perp & \mb{0}_{N-1}\\
\mb{0}_{N-1}^T & \zeta
\end{pmatrix},
\end{equation}
where $\nabla_u$ is the gradient of $\wt{G}$ with respect to $\mb{u}$
while keeping $b_N$ fixed.
We have thus obtained the requisite $N$-dimensional ODE system in the 
variables $\mb{u}$; the electroneutrality constraint $b_N=\text{constant}$ 
only appears as a parameter of the system.

Let $H_u$ be the Hessian of $\wt{G}$ with respect to $\mb{u}$. 
Given Lemma \ref{propG}, $H_u$ is symmetric positive definite.
Indeed, the quadratic form defined by $H_u$ is just the restriction 
of $H_G$ (defined in Lemma \ref{propG}) to the subspace of $\mathbb{R}^{N+1}$
orthogonal to $(\mb{z}^T,0)^T$.
The linearized operator of \eqref{ELE} around 
steady state is thus given by $-L_uH_u^*$ where $H_u^*$
is the evaluation of $H_u$ at the steady state. 
Note that $L_uH_u^*$ is similar 
to $(H_u^*)^{1/2}L_u(H_u^*)^{1/2}$ where $(H_u^*)^{1/2}$ is the 
positive square root of $H_u^*$, which exists thanks to positive 
definiteness of $H_u^*$. Since $(H_u^*)^{1/2}L_u(H_u^*)^{1/2}$
is a symmetric matrix and since $L_u$
is symmetric positive definite, so is 
$(H_u^*)^{1/2}L_u(H_u^*)^{1/2}$. Thus, $-L_uH_u^*$ is diagonalizable 
with real negative eigenvalues. The steady state is asymptotically 
stable and the approach to steady state is exponential.
\end{proof}
By rewriting \eqref{linsys} as \eqref{ELE}, we see that
the system is a gradient flow on the hyperplane defined by 
the electroneutrality constraint where the metric is given by $L_u^{-1}$.
This led us to the conclusion that the linearization is 
diagonalizable with real negative eigenvalues. Note that 
we made essential use of the {\em symmetry} of the matrix $L$, 
which came from the Onsager reciprocity principle (see \eqref{posdef}).
The gradient structure of our system combined with the convexity 
of $\wt{G}$ has another interesting consequence as we shall see 
in Lemma \ref{propJ}.

\subsection{Global Behavior}\label{glbehavior}

We now state the main result of this Section.
\begin{theorem}\label{globalstab}
Suppose condition \eqref{fphimin} is satisfied. Then, 
\eqref{linsys} has a unique steady state and 
it is globally asymptotically stable.
The linearization around steady state is diagonalizable with 
real negative eigenvalues.
\end{theorem}
Existence of the unique steady state was proved in Proposition \ref{linextcond}.
Asymptotic stability and the property of the linearized 
operator was proved in Proposition \ref{linmain}.
We have thus only to prove that all solutions with initial 
value in $\mathcal{S}$ (see \eqref{icS}) 
converge to the steady state as $t\to \infty$.
Implicit in this assertion is that these solutions 
are global (defined for all positive time).
Once this is established, we use the fact that $\wh{G}$
is a Lyapunov function to obtain the 
desired result.
 
To prove that all solutions are global, 
we must rule out two possibilities.
The first is that the solution may grow unbounded 
in finite time. To show that this is not possible, 
we make use of the function $\wh{G}$.
The second possibility is that one or more of the 
concentrations $c_k$ or the cell volume $v$ may come
arbitrarily close to $0$ in finite time. 
To show that this cannot happen, we examine the free energy 
dissipation function $J$ defined in \eqref{GJdef}.
\begin{lemma}\label{propJ}
View $J$ defined in \eqref{GJdef} as a function of 
$(\mb{c},v)\in \mathbb{R}_+^N$:
\begin{equation}\label{Jcdef}
J(\mb{c},v)=J_c(\mb{c})+\zeta(\pi_{\rm w}(\mb{c},v))^2, \; 
J_c(\mb{c})=\ip{\bm{\gamma}+{\mb{q}}}{\wh{L}(\bm{\gamma}+\mb{q})}_{\mathbb{R}^N}.
\end{equation}
\begin{enumerate}
\item Consider any solution $(\mb{c}(t),v(t))$ of system \eqref{linsys}.
We have:
\begin{equation}\label{LyapJ}
\begin{split}
\D{}{t}J(\mb{c}(t),v(t))&=
-\frac{2}{v}\sum_{k=1}^N
\paren{\frac{\rho_k}{\sqrt{c_k}}-\zeta \pi_{\rm w}\sqrt{c_k}}^2
-\frac{2A}{v^2}(\zeta\pi_{\rm w})^2,\\
\bm{\rho}&=(\rho_1,\cdots,\rho_N)^T=\wh{L}(\bm{\gamma}+\mb{q}).
\end{split}
\end{equation}
The function $J$ is thus a Lyapunov function in the sense that 
it is monotone non-increasing.
\item The function $J_c(\mb{c})$ defined in \eqref{Jcdef} 
tends to $+\infty$ as $\mb{c}$ approaches 
any point on $\partial \mathbb{R}^N_+$ 
where $\partial \cdot$ denotes the boundary of a set.
\end{enumerate}
\end{lemma}
That $J$ is a Lyapunov function can be seen as follows.
View $J$ as a function of $\mb{u}$ introduced in the proof of 
Proposition \ref{linmain}. First, note that:
\begin{equation}
\D{\wt{G}}{t}
=-\ip{\nabla_u \wt{G}}{\D{\mb{u}}{t}}_{\mathbb{R}^N}
=-\ip{\nabla_u \wt{G}}{L_u\nabla_u\wt{G}}_{\mathbb{R}^N}=-J
\end{equation}
where we used \eqref{ELE} in the second equality. We thus have:
\begin{equation}\label{DtJconvex}
\begin{split}
\D{J}{t}&
=2\ip{L_u\nabla_u \wt{G}}{\D{}{t}\paren{\nabla_u \wt{G}}}_{\mathbb{R}^N}
=2\ip{L_u\nabla_u \wt{G}}{H_u\D{\mb{u}}{t}}_{\mathbb{R}^N}\\
&=-2\ip{L_u\nabla_u \wt{G}}{H_u L_u\nabla_u \wt{G}}_{\mathbb{R}^N}
\end{split}
\end{equation}
where $H_u$ is the Hessian matrix of $\wt{G}$ with respect to $\mb{u}$.
We used \eqref{ELE} in the last equality.
As we saw in the proof of Proposition \ref{linmain},
$H_u$ is positive definite. Therefore, $J$ is monotone non-increasing.
We see that the Lyapunov property of $J$ is a 
general consequence of the fact our system, in suitable variables, 
is a gradient flow of the convex energy function $\wt{G}$.
\begin{proof}[Proof of Lemma \ref{propJ}]
We saw above that $dJ/dt$ is non-positive, 
but we have not obtained the right hand expression in \eqref{LyapJ}.
This is most easily done by direct calculation.
We turn to the second claim.  
Take any point $\mb{c}^{\rm b}\in \partial \mathbb{R}^N_+$ 
and assume without loss 
of generality that the first $1\leq l\leq N$ components of $\mb{c}^{\rm b}$ are $0$:
\begin{equation}
\mb{c}^{\rm b}=(\underbrace{0,\cdots,0}_{l},
c_{l+1}^{\rm b},\cdots,c_N^{\rm b}).
\end{equation}
Decompose the vector $\bm{\gamma}+\mb{q}$ in the following fashion.
\begin{equation}\label{gamma12def}
\begin{split}
\bm{\gamma}+\mb{q}&=-(\bm{\gamma}_1+\bm{\gamma}_2),\\
\bm{\gamma}_1&=-(\gamma_1+q_1,\cdots,\gamma_l+q_l,\underbrace{0,\cdots,0}_{N-l})^T,\\
\quad
\bm{\gamma}_2&=-(\underbrace{0,\cdots,0}_l,\gamma_{l+1}+q_{l+1},\cdots,\gamma_{N}+q_N)^T.
\end{split}
\end{equation}
Now, consider a sequence of points $\mb{c}^n\in \mathbb{R}^N_+, n=1,2,\cdots$ such that 
$\mb{c}^n\to \mb{c}^{\rm b}$ as $n\to \infty$.
Given that:
\begin{equation}
\gamma_k+q_k=\ln \paren{\frac{c_k}{c_k^{\rm e}}}+q_k, \; k=1,\cdots,N,
\end{equation}
each of the first $l$ non-zero components of $\bm{\gamma}_1$ goes to $+\infty$
as $\mb{c}^n\to \mb{c}^{\rm b}$ whereas $\bm{\gamma}_2$
remains bounded as $\mb{c}^n\to \mb{c}^{\rm b}$. 
Now, take an arbitrary vector 
\begin{equation}
\mb{w}=(w_1,\cdots,w_N)\in \mathbb{R}^N, \; \abs{\mb{w}}=1, \; w_k\geq 0,\;  k=1,\cdots,N.
\end{equation}
Recall from \eqref{whatLw} and the subsequent discussion
that $\ip{\mb{w}}{\wh{L}\mb{w}}_{\mathbb{R}^N}$ is positive if $\mb{w}$
is not parallel to $\mb{z}$. The vectors $\mb{w}$ and $\mb{z}$
are indeed not parallel since $\mb{z}$ must have at least 
one component that is negative (there is at least one ionic 
species with negative valence). Therefore,
\begin{equation}
\min_{w_k\geq 0,k=1,\cdots,N, 
\abs{\mb{w}}=1}\ip{\mb{w}}{\wh{L}\mb{w}}_{\mathbb{R}^N}\equiv K_w>0 
\end{equation}
given that the set satisfying $w_k\geq 0, k=1,\cdots,N, \abs{\mb{w}}=1$ is compact.
Therefore, for any vector $\mb{u}$ whose components are non-negative, we have:
\begin{equation}\label{Kwineq}
\ip{\mb{u}}{\wh{L}\mb{u}}_{\mathbb{R}^N}\geq K_w\abs{\mb{u}}^2.
\end{equation}
Now, let us take the limit of $J_c(\mb{c})$ as $\mb{c}^n \to \mb{c}^{\rm b}$.
If $\mb{c}^n$ is sufficiently close to $\mb{c}^{\rm b}$, the first 
$l$ components of $\bm{\gamma}_1$ as defined in \eqref{gamma12def} 
are positive.
Therefore, we have:
\begin{equation}
\begin{split}
J_c(\mb{c}^n)&\geq \ip{\bm{\gamma}_1+\bm{\gamma}_2}{\wh{L}(\bm{\gamma}_1+\bm{\gamma}_2)}_{\mathbb{R}^N}\\
&\geq \paren{K_w\abs{\bm{\gamma}_1}-2\abs{\wh{L}\bm{\gamma}_2}}\abs{\bm{\gamma}_1}
+\ip{\bm{\gamma}_2}{\wh{L}\bm{\gamma}_2}_{\mathbb{R}^N}
\end{split}
\end{equation}
where we used \eqref{Kwineq} and the Cauchy-Schwarz inequality. 
Since $\abs{\bm{\gamma}_1}\to +\infty$ and $\bm{\gamma}_2$ remains bounded as 
$\mb{c}^n\to \mb{c}^{\rm b}$, $J_c(\mb{c}^n)$ tends to $+\infty$.
\end{proof}

\begin{proof}[Proof of Theorem \ref{globalstab}]
Take an arbitrary initial value $(\mb{c}^0,v^0)\in \mathcal{S}$
where $\mathcal{S}$ was defined in \eqref{icS}.
We first show that the solution to \eqref{linsys} starting from 
$(\mb{c}^0,v^0)$ is defined for all $t>0$.

View $\wh{G}$ of \eqref{hatG} as a function of $(\mb{c},v)$.
Consider the set:
\begin{equation}\label{Ac}
\mathcal{A}_M=\lbrace (\mb{c},v)\in \mathcal{S}|
\wh{G}(\mb{c},v)< M \rbrace,
\end{equation}
where we choose $M$ so that $M>\wh{G}(\mb{c}^0,v^0)$.
Given \eqref{Ghateq}, the solution stays 
within $\mathcal{A}_M$ so long as the solution is defined.
We first show that the set $\mathcal{A}_M$ is bounded and 
that it is bounded away from the hyperplane $v=0$. 

Any element in $\mathcal{A}_M$ satisfies:
\begin{equation}\label{Mexp}
\sum_{k=1}^N \paren{c_k\paren{\ln \paren{\frac{c_k}{c_k^*}}-1}
+c_k^*}
<-\frac{A}{v}\paren{\ln \paren{\frac{v^*}{v}}-1}-\frac{A}{v^*}+\frac{M}{v}.
\end{equation}
It is easily seen that the left hand side of the above is greater 
than or equal to $0$. Therefore the right hand side must be 
greater than $0$, from which we obtain:
\begin{equation}
A\paren{\ln\paren{\frac{v^*}{v}}-1+\frac{v}{v^*}}<M.
\end{equation}
Therefore, $v$ must satisfy
$0<v_-<v<v_+<\infty$
for some constants $v_+$ and $v_-$.
Let $\overline{M}$ be the 
supremum of the right hand side of \eqref{Mexp} over 
$v_-<v<v_+$. This $\overline{M}$ is clearly finite. Thus, 
\begin{equation}
\sum_{k=1}^N \paren{c_k\paren{\ln \paren{\frac{c_k}{c_k^*}}-1}
+c_k^*}\leq \overline{M}.
\end{equation} 
We thus see that $c_k$ must be bounded above by a constant 
$c_+$ that depends only on $M$. Therefore, we have:
\begin{equation}\label{vbound}
0<v_-<v<v_+, \quad c_k<c_+, k=1,\cdots, N.
\end{equation}

Let $(\mb{c}(t),v(t))$ be the 
solution to \eqref{linsys} with initial data $(\mb{c}^0,v^0)$.
Since the solution stays within $\mathcal{A}_M$, we know that 
the solution satisfies the bound \eqref{vbound}.
We now show that the concentrations $c_k(t)$ are bounded away 
from $0$. Recall from \eqref{LyapJ} of Lemma \ref{propJ} that the function 
$J$ is a non-increasing function in time. 
Thus, $J(\mb{c}(t),v(t))\leq J(\mb{c}^0,v^0)=M_J$.
Since $J_c(\mb{c})\leq J(\mb{c},v)$, we have $J_c(\mb{c}(t))\leq M_J$.
By the second item in Lemma \ref{propJ}, the set:
\begin{equation}
\lbrace \mb{c}=(c_1,\cdots,c_N)\in \mathbb{R}_+^N| J_c(\mb{c})\leq M_J, 
c_k<c_+, k=1,\cdots N\rbrace
\end{equation}
must be bounded away from $\partial\mathbb{R}_+^N$. Therefore, we have:
\begin{equation}
0<c_-<c_k(t)<c_+, k=1,\cdots, N,
\end{equation}
where $c_-$ is a constant that depends only on $M$ and $M_J$.
This, together with \eqref{vbound}, implies that the solution 
$(\mb{c}(t),v(t))$ lies in a compact subset $\mathcal{K}$ of 
$\mathcal{S}$.
This shows that the solution must be defined for all time. 

We now show that the solution $(\mb{c}(t),v(t))$ converges to the 
steady state $(\mb{c}^*,v^*)$. Take an arbitrary $\epsilon>0$
and let $\mathcal{B}_\epsilon\subset\mathbb{R}^{N+1}$ be the open ball of radius $\epsilon$
centered at $(\mb{c}^*,v^*)$.
We must show that 
$(\mb{c}(t),v(t))\in \mathcal{B}_\epsilon$ after finite time. 
Observe that we can make $\delta>0$ sufficiently small 
so that $\mathcal{A}_\delta\subset \mathcal{B}_\epsilon$
($\mathcal{A}_\delta$ is defined by replacing $M$ with $\delta$
in \eqref{Ac}).  
This is clear since, by Lemma \ref{propG},  
$(\mb{c}^*,v^*)$ is the unique minimizer 
of $\wh{G}$ over $\overline{\mathcal{S}}$.

Take $\delta$ so small that $\mathcal{A}_\delta\in \mathcal{B}_\epsilon$.
If $M\leq \delta$, $\mathcal{A}_M\subset \mathcal{A}_\delta\subset \mathcal{B}_\epsilon$. 
Since the solution is contained in $\mathcal{A}_M$, there is
nothing to prove. Assume $M>\delta$. We would like to show
that the solution is contained in $\mathcal{A}_\delta \subset \mathcal{B}_\epsilon$
after finite time. We prove this by contradiction. Suppose the solution 
never enters $\mathcal{A}_\delta$. Recall that the solution 
$(\mb{c}(t),v(t))$ was 
contained in a compact set $\mathcal{K}\subset\mathcal{S}$.
The function
$J(\mb{c},v)$ is clearly positive on $\mathcal{K}\backslash\mathcal{A}_\delta$,
since $(\mb{c}^*,v^*)\in \mathcal{A}_\delta$ is the only point at 
which $J=0$.
Since $\mathcal{K}\backslash\mathcal{A}_\delta$ is a compact set, 
$J>K_J>0$ on $\mathcal{K}\backslash\mathcal{A}_\delta$ where $K_J$ is a 
positive constant.
By \eqref{GJdef} of Lemma \ref{FEq} (or equivalently, \eqref{Ghateq}),
we see that:
\begin{equation}
\wh{G}(\mb{c}(t),v(t))< M-K_Jt.
\end{equation}
This implies that the solution will be in the set 
$\mathcal{A}_\delta$ for $t>(M-\delta)/K_J$, a contradiction.
\end{proof}
Theorem \ref{globalstab} shows that system \eqref{linsys} has 
the following remarkable robustness property. Suppose the pump rates 
$\mb{p}$ and the extracellular concentrations $\mb{c}^{\rm e}$
are perturbed within the bounds of condition \eqref{fphimin}.
Then, the cell will approach the new global steady state.  

The next theorem shows that 
when condition \eqref{fphimin} is not met, 
the cell volume $v(t)\to \infty$ as $t\to \infty$.
There is thus a dichotomy in the behavior 
of system \eqref{linsys} depending on whether condition 
\eqref{fphimin} is satisfied.
\begin{theorem}\label{burst}
Suppose condition \eqref{fphimin} is not satisfied so that 
$f_{\rm min}(\mb{q},\mb{c}^{\rm e},\mb{z})\geq 0$.
Take any solution $(\mb{c}(t),v(t))$ to \eqref{linsys}.
\begin{enumerate}
\item Suppose $f_{\rm min}(\mb{q},\mb{c}^{\rm e},\mb{z})>0$.
Then, 
\begin{equation}\label{vtoinftyin1stcase}
\lim_{t \to \infty} v(t)=\infty.
\end{equation}
\item Suppose $f_{\rm min}(\mb{q},\mb{c}^{\rm e},\mb{z})=0$.
Let $\phi_{\rm min}$ be as in Proposition \ref{linextcond}.
Then,
\begin{equation}\label{vtoinfty2ndcase}
\begin{split}
\lim_{t\to \infty}c_k(t)&= c_k^{\rm e}\exp(-q_k-z_k\phi_{\rm min}), 
\; k=1,\cdots N,\\
\lim_{t\to \infty}v(t)&= \infty.
\end{split}
\end{equation}
\end{enumerate}
\end{theorem}
\begin{proof}
We shall work with the variables $\mb{c}$ and $w=1/v$.
In these variables, \eqref{linsys} can be written as:
\begin{subequations}\label{linsysincw}
\begin{align}
\D{\mb{c}}{t}&=w(-\wh{L}(\bm{\gamma}+\mb{q})+\zeta \pi_{\rm w}\mb{c}),
\label{linsysincw1}\\
0&=\sum_{k=1}^N z_kc_k+zAw=\sum_{k=1}^N z_kc_k^{\rm e},\label{linsysincwEN}\\
\D{w}{t}&=w^2\zeta \pi_{\rm w}.\label{linsysincw2}
\end{align}
\end{subequations}
The solutions are defined on the set:
\begin{equation}\label{defT}
\mathcal{T}=\lbrace (\mb{c},w)\in \mathbb{R}_+^{N+1}| 
\sum_{k=1}^N z_kc_k+zAw=0 \rbrace.
\end{equation}
Note that $\mathcal{T}$ is just 
the set $\mathcal{S}$ of \eqref{icS} written in 
the $(\mb{c},w)$ coordinates. 
Take any initial data $(\mb{c}^0,w^0)\in \mathcal{T}$
and let $(\mb{c}(t),w(t))$ be the solution to \eqref{linsysincw}
starting from this point. Showing that $v(t)\to \infty$ is 
equivalent to showing that $w(t)\to 0$.
We divide the proof into several steps.

{\em Step 1}:
View $\wt{G}$ defined in \eqref{GJdef} as a function of $(\mb{c},w)$.
Consider the set:
\begin{equation}
\mathcal{A}_M=\lbrace (\mb{c},w)\in \mathbb{R}_+^{N+1}| 
\sum_{k=1}^N z_kc_k +zAw=0, \;\wt{G}(\mb{c},w)< M \rbrace.
\end{equation}
This is the same set as \eqref{Ac} except that we use the function 
$\wt{G}(\mb{c},w)$ instead 
of $\wh{G}(\mb{c},v)$.
We prove that $\mathcal{A}_M$ is a bounded set.

For any point in $\mathcal{A}_M$, we have:
\begin{equation}\label{Ckvcke}
\sum_{k=1}^N \paren{c_k\paren{\ln\paren{\frac{c_k}{c_k^{\rm e}}}-1+q_k}+c_k^{\rm e}}
<Mw-Aw\paren{\ln(Aw)-1}.
\end{equation}
Since the left hand side is bounded 
from below, there is some positive constant $m$, 
independent of $M$, such that:
\begin{equation}\label{mMcond}
-m<Mw-Aw(\ln(Aw)-1).
\end{equation}
Let:
\begin{equation}
g(w)=-\frac{m}{w}+A(\ln(Aw)-1).
\end{equation}
The function $g(w)$ is a monotone increasing function in $w$ such that 
$g(w)\to -\infty$ as $w\to 0$ and $g(w)\to \infty$ as $w\to \infty$. 
Let $w_+(M)=g^{-1}(M)$. Given \eqref{mMcond}, we have:
\begin{equation}\label{v_-M}
0<w<w_+(M) \text{ for any } (\mb{c},w)\in \mathcal{A}_M.
\end{equation}
Since $w$ is bounded between $0$ and $w_+(M)$, we see from \eqref{Ckvcke}
that $c_k$ must also be bounded in $\mathcal{A}_M$:
\begin{equation}\label{ckBMbound}
0<c_k<c_+(M), \; k=1,\cdots N \text{ for any } (\mb{c},w)\in \mathcal{A}_M.
\end{equation}

{\em Step 2}:
We prove that the solution $(\mb{c}(t),w(t))$ is defined for all positive time.
Choose $M=M_0$ so that $M_0>\wt{G}(\mb{c}_0,w_0)$. 
Suppose that the solution exists only up to $t<T_0, T_0<\infty$.
Since $\wt{G}$ is monotone non-increasing, we have
$(\mb{c}(t),w(t))\in \mathcal{A}_{M_0}, t<T_0$. 
Since $\mathcal{A}_{M_0}$ is bounded by \eqref{v_-M} and \eqref{ckBMbound},
there is a sequence of times 
$t_1<t_2\cdots \to T_0$ such that 
$\mb{x}(t_n)=(\mb{c}(t_n),w(t_n))\to 
\mb{x}^{\rm b}=(\mb{c}^{\rm b},w^{\rm b})\in \partial \mathcal{A}_{M_0}$
as $n\to \infty$. 
The limit point $\mb{x}^{\rm b}$ cannot be in $\mathbb{R}_+^{N+1}$ since, 
if so, the solution can be continued beyond time $T_0$.
We also see that $c_k, k=1,\cdots,N$ must stay away from $0$
by an argument using Lemma \ref{propJ}
similarly to the proof of Theorem \ref{globalstab}.
This implies that $\mb{c}^{\rm b}\in \mathbb{R}_+^N$
and $w^{\rm b}=0$. By \eqref{linsys2}, we have:
\begin{equation}\label{bound1/w}
\D{}{t}\paren{\frac{1}{w}}=\zeta\paren{\sum_{k=1}^N \paren{c_k-c_k^{\rm e}}+Aw}.
\end{equation}
Since the right hand side of the above is bounded in $\mathcal{A}_{M_0}$
by \eqref{v_-M} and \eqref{ckBMbound}, 
$1/w(t)$ remains finite in finite time. Thus, $w(t_n)\to 0$ is impossible 
as $t_n \to T_0<\infty$. 
We have a contradiction.

{\em Step 3}: 
Let $\mathcal{O}$ be the orbit:
\begin{equation}
\mathcal{O}=\lbrace (\mb{c}(t),w(t))\in \mathcal{T}, t\geq 0\rbrace.
\end{equation}
We would like to see whether 
\begin{equation}
\underline{J}=\inf_{(\mb{c},w)\in \mathcal{O}} J(\mb{c},w)
\end{equation}
is positive. 
Recall that $J=0$ in $\mathcal{T}$
if and only if the point $(\mb{c},w)$ is a steady state of \eqref{linsys}. 
Given our assumption that \eqref{fphimin} is not satisfied, by 
Proposition \ref{linextcond}, a steady state does not exist. 
Therefore, $J>0$ in $\mathcal{O}\subset \mathcal{T}$. 
Thus, if $\underline{J}=0$, 
since $\mathcal{O}$ is a bounded set,
there is a sequence of points $\mb{x}^n\in \mathcal{O}, n=1,2,\cdots$
that approaches a point $\mb{x}^\infty=(\mb{c}^\infty,w^\infty)
\in \overline{\mathcal{T}}$
such that $J(\mb{x}^n)\to 0$ as $n\to \infty$. The limit point $\mb{x}^\infty$
cannot be in $\mathcal{T}\in \mathbb{R}_+^{N+1}$ since $J>0$ there. Since 
$\mb{c}(t)$ stays away from $\partial \mathbb{R}_+^{N}$,
$\mb{c}^\infty\notin\partial\mathbb{R}_+^N$. Thus,
$\mb{c}^\infty\in \mathbb{R}_+^N$ and $w^\infty=0$. 
Since the function $J$ is continuous up to points 
$(\mb{c},w)=(\mb{c}^\infty,0), \mb{c}^\infty\in \mathbb{R}_+^N$, 
we examine the positivity of $J$ on the set:
\begin{equation}
\mathcal{R}=
\lbrace\mb{x}=(\mb{c},w)\in \partial{\mathcal{T}}
|\mb{c}\in \mathbb{R}_+^N,\; w=0
\rbrace. 
\end{equation}
On $\mathcal{R}$, $J$ can be written as:
\begin{equation}\label{JPsipi}
J(\mb{c},w=0)=J_c(\mb{c})+\zeta(\pi_{\rm w}^0(\mb{c}))^2, \; 
\pi_{\rm w}^0(\mb{c})=\sum_{k=1}^N \paren{c_k^{\rm e}-c_k}.
\end{equation}
We see that $J=0$ if and only if $J_c(\mb{c})=0$ and $\pi_{\rm w}^0=0$.
It is easily seen that $J_c(\mb{c})=0$ in $\mathcal{R}$
if and only if $\mb{c}=\mb{c}^*$ where $\mb{c}^*$ is given by:
\begin{equation}\label{ckqkzkphimin}
\mb{c}^*=(c_1^*,\cdots,c_N^*)^T,\;
c_k^*\equiv c_k^{\rm e}\exp(-q_k-z_k\phi_{\rm min}), \; k=1,\cdots N, 
\end{equation} 
where $\phi_{\rm min}$ is as defined in the statement of Proposition 
\ref{linextcond}. Let us evaluate 
$\pi_{\rm w}^0$ at this point. Substituting \eqref{ckqkzkphimin} into 
\eqref{JPsipi} and recalling the definition of $f_{\rm min}$ in \eqref{fphimin},
\begin{equation}
\pi_{\rm w}^0(\mb{c}^*)=-f_{\rm min}(\mb{q},\mb{c}^{\rm e},\mb{z}).
\end{equation}
Therefore, we have:
\begin{equation}\label{JM>0}
 \underline{J}>0 \text{ if } f_{\rm min}>0.
\end{equation}
When $f_{\rm min}=0$, we have the following.
Let $\mathcal{B}_\eta$ be the open ball 
of radius $\eta>0$ centered at $(\mb{c}^*,0)$.
Then,
\begin{equation}\label{KJ}
\underline{J}_\eta \equiv \inf_{(\mb{c},w)\in \mathcal{O}\backslash \mathcal{B}_\eta}
J(\mb{c},w)>0.
\end{equation}

{\em Step 4}: We prove our claim when $f_{\rm min}>0$.
Given that $(\mb{c}(t),w(t))\in \mathcal{O}$, by Lemma \ref{FEq} 
we have:
\begin{equation}
\wt{G}(\mb{c}(t),w(t))<M_0-\underline{J}t.
\end{equation}
By \eqref{v_-M}, we have:
\begin{equation}\label{wsandwich}
0<w(t)<w^+(M_0-\underline{J}t).
\end{equation}
Note that $\underline{J}>0$ by \eqref{JM>0}.
Since $w^+(M)\to 0$ as $M\to -\infty$, we see from \eqref{wsandwich}
that $w(t) \to 0$ as $t \to \infty$. 
This proves \eqref{vtoinftyin1stcase}.

{\em Step 5}: 
In the rest of the proof, we study the $f_{\rm min}=0$ case. 
As an initial step, we prove the following. 
For any $\eta>0$, there is a time $t_\eta\geq 0$
such that the point 
$(\mb{c}(t_\eta),w(t_\eta))$ is in $\mathcal{B}_\eta$.
We prove this by contradiction. Suppose otherwise.
Then, there is an $\eta>0$ such that 
$\mathcal{B}_\eta\cap \mathcal{O}$ is empty. 

We first show that $\wt{G}$ is bounded from below in 
$\mathcal{O}$ by a constant $M_\eta$.
Suppose otherwise. Then, 
there is a sequence of points $\mb{x}^n=(\mb{c}^n,w^n)\in 
\mathcal{O}, n=1,2,\cdots$ converging 
to $\mb{x}^\infty=(\mb{c}^\infty,w^\infty)\in 
\overline{\mathcal{O}}$ such that 
$\wt{G}(\mb{c}^n,w^n)\to -\infty$. Since $\wt{G}(\mb{c},w)$ is a continuous 
function for $\mb{c}\in \overline{\mathbb{R}_+}^N, w>0$, the
only possibility is that $w^\infty=0$.
Write $\wt{G}$ as:
\begin{equation}\label{writeGrho}
\begin{split}
\wt{G}(\mb{c},w)&=\frac{1}{w}\rho(\mb{c})+A(\ln(Aw)-1), \\
\rho(\mb{c})&=\sum_{k=1}^N
\paren{c_k\paren{\ln \paren{\frac{c_k}{c_k^{\rm e}}}-1+q_k}+c_k^{\rm e}}.
\end{split}
\end{equation}
It suffices to show that $\rho(\mb{c}^\infty)>0$. If this is true, we see 
from \eqref{writeGrho} that $\wt{G}(\mb{c}^n,w^n)\to \infty$,
contradicting our assumption that $\wt{G}(\mb{c}^n,w^n)\to -\infty$.
It is easily seen by a calculation identical to the proof of Lemma 
\ref{propG} that the unique minimizer of $\rho(\mb{c})$ under 
the constraint 
\begin{equation}\label{zkckconstraint}
\sum_{k=1}^N z_kc_k=0.
\end{equation} 
is $\mb{c}=\mb{c}^*$, at 
which point $\rho(\mb{c}^*)=0$. Given that 
$(\mb{c}^\infty,0)\notin \mathcal{B}_\eta$, 
we see that $\mb{c}^\infty\neq \mb{c}^*$,
and thus $\rho(\mb{c}^\infty)>0$.

By Lemma \ref{FEq} and using \eqref{KJ}, we have:
\begin{equation}
\wt{G}(\mb{c}(t),w(t))\leq M_0-\underline{J}_\eta t 
\end{equation}
so long as $(\mb{c}(t),w(t))\notin \mathcal{B}_\eta$. 
Note that $\underline{J}_\eta>0$ by \eqref{KJ}.
Thus, if $t>(M_0-M_\eta)/\underline{J}_\eta$, then
$\wt{G}<M_\eta$, which contradicts our result that 
$\wt{G}$ must be greater than $M_\eta$ on $\mathcal{O}$.

{\em Step 6}:
We would like to show that there is a positive number $\eta>0$
such that any solution with initial data in $\mathcal{B}_\eta\cap \mathcal{T}$ 
will converge to $(\mb{c}^*,0)$ as $t\to \infty$. 
If this is true, we can combine this with the result of 
Step 5 to immediately conclude that 
that all solutions of \eqref{linsysincw} converge to 
$(\mb{c},w)=(\mb{c}^*,0)$ 
as $t\to \infty$. This would prove \eqref{vtoinfty2ndcase}.

The vector field defined by the 
right hand sides of \eqref{linsysincw1} and \eqref{linsysincw2} 
is degenerate at $w=0$. Rescaling the vector field
by a positive scalar factor 
does not alter the solution orbits, so we shall study the 
behavior of an appropriately rescaled system 
\citep{chicone1999,benson2010general}.  
Rescale \eqref{linsysincw1} and \eqref{linsysincw2} by 
a factor of $1/w$. This removes the degeneracy at $w=0$:
\begin{subequations}\label{linsysmod}
\begin{align}
\D{\mb{c}}{\tau}&=-\wh{L}(\bm{\gamma}+\mb{q})+\zeta \pi_{\rm w}\mb{c},\\
\D{w}{\tau}&=w\zeta \pi_{\rm w}.
\end{align}
\end{subequations}
We have taken the time parameter to be $\tau$ to distinguish the 
solutions of this system with those of \eqref{linsysincw}.
We consider \eqref{linsysmod} on the set:
\begin{equation}
\mathcal{T}'=\lbrace (\mb{c},w)\in \mathbb{R}_+^{N}\times 
\overline{\mathbb{R}_+}| 
\sum_{k=1}^N z_kc_k+zAw=0 \rbrace.
\end{equation}
The difference between $\mathcal{T}$ and $\mathcal{T}'$ is whether or not
$w$ is allowed to be equal to $0$.

If we can show that any solution to \eqref{linsysmod} with 
initial data in $\mathcal{B}_\eta\cap \mathcal{T}$ 
converges to $(\mb{c}^*,0)$ as $\tau\to \infty$, the 
same is true for \eqref{linsysincw} as $t\to \infty$. 
This can be seen as follows.
Let $(\mb{c}(t),w(t))$ be the solution to \eqref{linsysincw} with 
initial data $(\mb{c}^0,w^0)\in \mathcal{B}_\eta\cap \mathcal{T}$
and $(\wt{\mb{c}}(\tau),\wt{w}(\tau))$ be the solution to  \eqref{linsysmod}
with the same initial conditions. Define the function $\xi(t)$ with:
\begin{equation}\label{xidef}
t=\int_0^{\xi(t)}\frac{1}{\wt{w}(\tau)}d\tau.
\end{equation}
This function is well-defined since $\wt{w}(\tau)>0$. This positivity 
is a simple consequence of the backward uniqueness of solutions. 
The expressions
$\wt{\mb{c}}(\xi(t))$ and $\wt{w}(\xi(t))$ satisfy \eqref{linsysincw}, and 
thus, by uniqueness of solutions:
\begin{equation}
(\mb{c}(t),w(t))=(\wt{\mb{c}}(\xi(t)),\wt{w}(\xi(t))).
\end{equation}
By \eqref{xidef} and the fact 
that $\wt{w}(\tau) \to 0$ as $\tau \to \infty$, 
we see that $\xi \to \infty$ whenever $t\to \infty$. Therefore, 
\begin{equation}
\lim_{t\to \infty}(\mb{c}(t),w(t))=
\lim_{\tau \to \infty} (\wt{\mb{c}}(\tau),\wt{w}(\tau)).
\end{equation}

{\em Step 7}:
We now show that any solution of \eqref{linsysmod} with initial 
conditions in $\mathcal{B}_\eta\cap \mathcal{T}'$ converges to $(\mb{c}^*,0)$
if $\eta$ is taken small enough. 
This will conclude the proof. 

Let $\mb{c}(\tau),w(\tau)$ be a solution to \eqref{linsysmod}. 
Then, by Lemma \ref{propJ}, we have:
\begin{equation}\label{LyapJmod}
\D{}{\tau}J(\mb{c},w)=-2\sum_{k=1}^N
\paren{\frac{\rho_k}{\sqrt{c_k}}-\zeta \pi_{\rm w}\sqrt{c_k}}^2
-2Aw(\zeta\pi_{\rm w})^2\equiv \Psi(\mb{c},w).
\end{equation}
Since we have removed one factor of $w$ in \eqref{linsysmod} 
compared with \eqref{linsysincw}, one factor of $w$ is removed 
accordingly from the right hand side of \eqref{LyapJ}.

First, we take $r>0$ small enough so that 
$J>0$ and $\Psi>0$ for
$\overline{\mathcal{B}_{r}\cap \mathcal{T}'}\backslash\{(\mb{c}^*,0)\}$.
By \eqref{KJ}, this is possible for $J$. 
Consider $\Psi$. If $w>0$, 
the condition $\Psi=0$ if and only if $\rho_k=\pi_{\rm w}=0$. This 
is equivalent to the condition for \eqref{linsysincw}
to have a steady state in $\mathcal{T}$. By Proposition \ref{linextcond},
such a point does not exist if $f_{\rm min}=0$. For $w=0$, 
the $\Psi(\mb{c},0)=0$ if and only if:
\begin{align}\label{ststatecond}
\wh{L}(\bm{\gamma}+\mb{q})-\zeta \pi_{\rm w}^0\mb{c}&=0,\\
\sum_{k=1}^N z_kc_k&=0,\label{ststateEN}
\end{align}
where $\pi_{\rm w}^0$ was given in \eqref{JPsipi}.
It is clear that the point $\mb{c}=\mb{c}^*$ satisfies the above.
We show that $\mb{c}=\mb{c}^*$ is the only point that 
satisfies both \eqref{ststatecond} and \eqref{ststateEN}
in a neighborhood of $\mb{c}=\mb{c}^*$ in $\mathbb{R}^N$.
We use the implicit 
function theorem. Compute 
the Jacobian matrix of the left hand side with respect to 
$\mb{c}$ and evaluate this at $\mb{c}^*$:
\begin{equation}
B^*_{kl}=\wh{L}_{kl}\frac{1}{c^*_l}+\zeta c^*_k.
\end{equation}
Here, $B^*_{kl}$ is the $kl$ entry of the $N\times N$ Jacobian matrix $B^*$.
The rank of $B^*$ is the same as the rank of
$\wt{B}^*$ whose $kl$ entry is given by:
\begin{equation}
\wt{B}^*_{kl}=\wh{L}_{kl}+\zeta c^*_kc^*_l.
\end{equation}
Since $\wh{L}$ is symmetric positive semidefinite with rank $N-1$ (see 
proof of Proposition \ref{linmain}) and $\zeta>0$, $\wt{B}^*$, 
and thus $B^*$ is at least rank $N-1$. It is easily checked 
that $\mb{z}_c=(z_1c_1^*,\cdots,z_Nc_N^*)^T$ 
is an eigenvector of $B^*$ with $0$ eigenvalue.
Therefore, all the 
points that satisfy \eqref{ststatecond} near $\mb{c}=\mb{c}^*$
lie on a one-dimensional manifold in $\mathbb{R}^N$ that is 
tangent to $\mb{z}_c$ at $\mb{c}=\mb{c}^*$. 
Since $\ip{\mb{z}}{\mb{z}_c}_{\mathbb{R}^N}\neq 0$,  
the only common point between this one dimensional 
manifold and the hyperplane \eqref{ststateEN} is $\mb{c}=\mb{c}^*$.

Define the set:
\begin{equation}
\mathcal{D}_\delta=\lbrace (\mb{c},v)\in \mathcal{B}_r\cap \mathcal{T}'| 
J(\mb{c},v)<\delta \rbrace.
\end{equation}
Take $\delta=\delta_0$ small enough so that 
$\overline{\mathcal{D}_{\delta_0}}\subset \mathcal{B}_r$. It is clear by 
\eqref{LyapJmod} that any solution in $\mathcal{D}_{\delta_0}$ stays within 
this set. Choose $\eta$ small enough so that 
$\mathcal{B}_\eta\cap\mathcal{T}' 
\subset \mathcal{D}_{\delta_0}$. This is clearly possible since 
$J$ is a non-negative continuous function on 
$\mathcal{B}_r\cap\mathcal{T}'$ which is $0$
only at $(\mb{c},w)=(\mb{c}^*,0)$. Take any $\epsilon<\eta$. 
We may choose a $\delta_1>0$ so that $\mathcal{D}_{\delta_1}\subset \mathcal{B}_\epsilon$.
Given that $\Psi>0$ on 
$\overline{\mathcal{D}_{\delta_0}\backslash\mathcal{D}_{\delta_1}}$, 
any solution in $\mathcal{B}_\eta\subset \mathcal{D}_{\delta_0}$ 
will be in $\mathcal{B}_\epsilon\supset \mathcal{D}_{\delta_1}$ in finite time.
\end{proof}
In the case of $f_{\rm min}>0$, we do not have a statement on the 
limiting value of $\mb{c}(t)$ as $t\to \infty$. 
The limit point 
$(\mb{c},w)=(\mb{c}^*,0)$ may well exist and 
the limiting value $\mb{c}^*$ should 
satisfy \eqref{ststatecond}. This follows simply by setting 
the right hand side of \eqref{linsysmod} to $0$. If there 
is only one such point, it should be possible to show, with the aid of the 
Lyapunov function $J$, that this is the single limit 
point to which all solutions converge. This uniqueness, 
however, is not clear.

Even though $v(t)\to \infty$ as $t\to \infty$
regardless of whether $f_{\rm min}>0$ or $f_{\rm min}=0$, the rate at which 
$v(t)$ grows is different. When $f_{\rm min}>0$, 
we see, by combining \eqref{bound1/w} and \eqref{wsandwich}
and the definition of $w_+(M)$ that $v(t)$
grows at most linearly with time and faster than any power 
$t^\alpha, \alpha<1$. When $f_{\rm min}=0$, we expect the 
growth of cell volume to scale like $t^{1/2}$, as can be seen by taking the 
special case $\mb{p}=0$ with initial conditions $\mb{c}=\mb{c}^{\rm e}$.

\subsection{Epithelial Models}\label{epimod}

We briefly consider a simple epithelial model. Suppose we have 
a single layer of epithelial cells which separate the serosal and mucosal 
sides. The concentrations of electrolytes in the serosal and mucosal 
solutions are assumed constant. Let these concentrations be denoted 
$c_k^{\rm s}$ and $c_k^{\rm m}$ respectively. The voltages 
in the mucosal and serosal sides are also fixed at $\phi^{\rm s}$
and $\phi^{\rm m}$. We can write down the following model 
for electrolyte and water balance for an epithelial cell in this layer:
\begin{subequations}\label{epi}
\begin{align}
\D{(v\mb{c})}{t}&
=-L_{\rm m}\bm{\mu}^{\rm m}-L_{\rm s}\bm{\mu}^{\rm s}
-\mb{p}^{\rm m}-\mb{p}^{\rm s},\\
\D{v}{t}&=-\zeta_{\rm m}\pi_{\rm w}^{\rm m}-\zeta_{\rm s}\pi_{\rm w}^{\rm s},\\
\sum_{k=1}^N z_kc_k+\frac{zA}{v}&=\sum_{k=1}^N z_kc_k^{\rm m}
=\sum_{k=1}^N z_kc_k^{\rm s}=0.
\end{align}
\end{subequations}
The definition of the cellular variables $\mb{c}$ 
and $v$ are the same as before.
In the above, $L_{\rm m,s}$ are symmetric positive 
semi-definite matrices and $\mb{p}^{\rm m,s}$ 
are the vector of active currents residing on the mucosal 
and serosal membrane respectively which we assume constant. 
The chemical potentials 
$\bm{\mu}^{\rm m,s}=(\mu_1^{\rm m,s},\cdots,\mu_N^{\rm m,s})^T$ 
are given by:
\begin{equation}
\mu_k^{\rm m,s}=\ln \paren{\frac{c_k}{c_k^{\rm m,s}}}+z_k(\phi-\phi^{\rm m,s}),
\end{equation}
where $\phi$ is the electrostatic potential inside the cell.
The osmotic pressure $\pi_{\rm w}^{\rm m,s}$ is given by:
\begin{equation}
\pi_{\rm w}^{\rm m,s}=\sum_{k=1}^N c_k^{\rm m,s}
-\paren{\sum_{k=1}^N c_k+\frac{A}{v}}
\end{equation}
and the hydraulic permeabilities $\zeta^{\rm m,s}$ are non-negative constants.

The above problem is in fact mathematically identical to 
system \eqref{linsys}. 
Let:
\begin{equation}\label{zetaL}
\zeta=\zeta_{\rm m}+\zeta_{\rm s}, \; L=L_{\rm m}+L_{\rm s},
\end{equation} 
and suppose that $\zeta>0$ and $L$ is symmetric positive definite.
Define:
\begin{equation}\label{ckepbeta}
\begin{split}
c_k^{\rm e}&=\zeta^{-1}\paren{\zeta_{\rm m}c_k^{\rm m}+\zeta_{\rm s}c_k^{\rm s}},\\
\mb{p}&=\mb{p}^{\rm m}+\mb{p}^{\rm s}-L_{\rm m}\bm{\beta}^{\rm m}-L_{\rm m}\bm{\beta}^{\rm s},\\
\bm{\beta}^{\rm m,s}&=(\beta_1^{\rm m,s},\cdots,\beta_N^{\rm m,s})^T, \; 
\beta^{\rm m,s}_k=\ln \paren{\frac{c_k^{\rm m,s}}{c_k^{\rm e}}}+z_k\phi^{\rm m,s}.
\end{split}
\end{equation}
Then, the triple $(\mb{c},v,\phi)$ satisfies \eqref{epi} if and only if 
it satisfies 
\eqref{linsys} with 
$\zeta,L,\mb{c}^{\rm e}=(c_1^{\rm e},\cdots,c_N^{\rm e})^T$ 
and $\mb{p}$ prescribed as in \eqref{zetaL}
and \eqref{ckepbeta}.
We thus have the following result.
\begin{theorem}
Consider system \eqref{epi}.
Suppose $L_{\rm m}+L_{\rm s}$ is symmetric positive definite and 
$\zeta_{\rm m}+\zeta_{\rm s}>0$. Define $f_{\rm min}$ as in 
\eqref{fphimin}
in which $\mb{q}=(q_1,\cdots,q_N)=L^{-1}\mb{p}$ 
and $c_k^{\rm e}, L$ and $\mb{p}$ are prescribed as in 
\eqref{zetaL} and \eqref{ckepbeta}. 
If $f_{\rm min}<0$, the conclusions of Theorem \ref{globalstab} hold.
If $f_{\rm min}\geq 0$, the conclusions of Theorem \ref{burst} hold.
\end{theorem}
Note that this epithelial model also enjoys the robustness property 
described after the end of the proof of Theorem \ref{globalstab}. 
We may argue that this is advantageous for an epithelial cell,
which must withstand large changes in extracellular ionic concentrations.

\section{Results in the General Case}\label{general}

In the previous Section, we assumed that the passive transmembrane 
ionic flux $j_k$ is linear in $\bm{\mu}$. 
In this Section, we establish results that are valid when we
only assume conditions \eqref{noflux}, \eqref{jwinc} and \eqref{posdef}
for $j_k$ and $j_{\rm w}$. In particular, this will apply to the 
case when the Goldman equation \eqref{jkGHK} is used for $j_k$.
We also relax the assumption that the pump rates $p_k$ be constant.
We consider the system:
\begin{subequations}\label{gensys}
\begin{align}
\D{}{t}(v\mb{c})&=-\mb{j}(\phi,\bm{\mu})-\alpha \mb{p}(\phi,\bm{\mu}),
\label{ionsgen}\\
0&=\sum_{k=1}^N z_kc_k+z\frac{A}{v}=\sum_{k=1}^N z_kc_k^{\rm e},
\label{ENgen}\\
\D{v}{t}&=-j_{\rm w}(\pi_{\rm w}).\label{volgen}
\end{align}
\end{subequations} 
The extracellular ionic concentrations $c_k^{\rm e}, k=1,\cdots,N$
and the amount of impermeable organic solute $A$ are positive.
We assume that $\mb{j},\mb{p}$ and $j_{\rm w}$ are $C^1$ functions 
of their arguments.
The only difference between this and system \eqref{system} is that 
we replaced $\mb{p}$ (or $p_k$) in \eqref{ionsdless} with 
$\alpha \mb{p}$ (or $\alpha p_k$) 
where $\alpha$ is a pump strength parameter. 
We shall find it useful to vary this parameter in the 
statements to follow.

\subsection{Solvability}

We first discuss what we mean by a solution to the initial value 
problem for \eqref{gensys}. 
Consider the two constraints \eqref{ENgen} and (the dimensionless 
form of) \eqref{ENI}, 
which we reproduce here for convenience:
\begin{align}\label{Qconstraint}
Q(\mb{c},v)&\equiv \sum_{k=1}^Nz_kc_k+\frac{zA}{v}=0,\\
I(\mb{c},\phi,\alpha)
&\equiv \sum_{k=1}^Nz_k(j_k(\mb{c},\phi)+\alpha p_k(\mb{c},\phi))=
\ip{\mb{z}}{\mb{j}+\alpha \mb{p}}_{\mathbb{R}^N}=0.
\end{align}
Define the following set:
\begin{equation}
\Gamma_\alpha=\lbrace 
\mb{y}=
(\mb{c},v,\phi)\in \mathbb{R}_+^N\times \mathbb{R}_+\times \mathbb{R}\;
| \;Q(\mb{c},v)=I_\alpha(\mb{c},\phi)=0\rbrace.
\end{equation}
We shall often omit the dependence of $I$ and $\Gamma$ on $\alpha$.

\begin{definition}
Let 
$\mb{y}^0=(\mb{c}^0,v^0,\phi^0)\in \Gamma$ where 
$\mb{c}^0=(c_1^0,\cdots,c_N^0)^T$.
Let $\mb{c}(t)=(c_1(t),\cdots,c_N(t))^T, \; v(t),\; t\geq 0$ 
be $C^1$ functions 
and $\phi(t), \; t\geq 0$ be a continuous function of $t$.
The function $\mb{y}(t)=(\mb{c}(t), v(t), \phi(t))$ is a solution 
to \eqref{gensys} with initial values 
$\mb{y}^0$ if it satisfies \eqref{gensys} and 
$\mb{y}(0)=\mb{y}^0$.
The solution may or may not be defined for all positive time.
\end{definition}
Since we are solving a differential algebraic system, we must 
specify initial conditions that satisfy the constraints.  
Note that we require $\phi(t)$ to be a continuous function of $t$. 

We have the following on the solvability of \eqref{gensys}.
\begin{lemma}\label{gensolexist}
Let $\mb{y}_0=(\mb{c}^0,v^0,\phi^0)\in \Gamma$ 
and suppose $\partial{I}/\partial{\phi}\neq 0$ at this point.
Let $\mathcal{B}_r
\subset\mathbb{R}^{N+2}$ 
be the open ball of radius $r$ centered at $\mb{y}_0$.
Then, there is a $K>0$ such that $\mathcal{B}_K$ has the following property.
\begin{enumerate}
\item  The set
$\mathcal{B}_K\cap \Gamma$ 
is an $N$-dimensional submanifold of $\mathbb{R}^{N+2}$. 
There is a $C^1$ function $\Phi$ such that any point
$\mb{y}\in \mathcal{B}_K\cap \Gamma$ can be written as 
$\mb{y}=(\mb{c},v,\Phi(\mb{c}))$.
\item Take any point $\mb{y}^1=(\mb{c}^1,v^1,\phi^1)\in 
\mathcal{B}_K\cap \Gamma$.
There is a unique solution $\mb{y}(t)=(\mb{c}(t),v(t),\phi(t))$
to \eqref{gensys} with initial values $\mb{y}^1$
for short times. For short times, 
$\phi(t)=\Phi(\mb{c}(t))$, and thus $\phi(t)$
is a $C^1$ function.
\end{enumerate}
\end{lemma}
\begin{proof}[Proof of Lemma \ref{gensolexist}]
The first item is a straightforward consequence of the implicit 
function theorem.
For the second item, substitute $\phi=\Phi(\mb{c})$ into \eqref{gensys}.
Solve this ODE with initial values $(\mb{c}^1,v^1)$ and 
let $\mb{c}(t)$ and $v(t)$ be the resulting solution. 
It is clear that $(\mb{c}(t),v(t), \Phi(\mb{c}(t)))$ is a solution 
to \eqref{gensys} with initial values $(\mb{c}^1,v^1,\phi^1)$ for short times.
This solution is the unique solution, since
$\phi(t)$ must be continuous and thus, must remain within $\mathcal{B}_K$
for short times. 
\end{proof}
The same statement clearly holds if we replace $\mathcal{B}_K$ with 
a neighborhood of $\mb{y}_0$.
Note that, when $z\neq 0$, any point in $\mathcal{B}_K\cap \Gamma$
can be written as $(\mb{c},V(\mb{c}),\Phi(\mb{c}))$ where $V(\mb{c})$ is found 
by solving $Q(\mb{c},v)=0$ for $v$. Thus, when $z\neq 0$, $\mb{c}$
serves as a local coordinate system on $\mathcal{B}_K\cap \Gamma$.

Given the structure conditions on $j_k$ we have the following 
solvability result.
\begin{proposition}\label{wellposed}
Suppose $j_k$ satisfies \eqref{noflux} and \eqref{posdef}.
Let
\begin{equation}
\mathcal{C}_r=\lbrace \mb{y}=(\mb{c},v,\phi)\in \mathbb{R}^{N+2}\; |\; 
\abs{\mb{c}-\mb{c}^{\rm e}}<r, \; \abs{\phi}<r\rbrace,
\end{equation} 
where $\abs{\cdot}$ for a vector in $\mathbb{R}^N$ denotes its 
Euclidean norm. 
There are positive constants $K_\alpha$ and $K$ with the 
following properties. 
\begin{enumerate}
\item Take any $\abs{\alpha}<K_\alpha$. 
The set
$\mathcal{C}_K\cap \Gamma_\alpha$ is an (unbounded) $N$-dimensional 
submanifold of $\mathbb{R}^{N+2}$ such that any 
$\mb{y}\in \mathcal{C}_K\cap \Gamma_\alpha$ can be written as
$\mb{y}=(\mb{c},v,\Phi(\mb{c},\alpha))$
where $\Phi(\mb{c},\alpha)$ is a $C^1$ function of $\mb{c}$ and $\alpha$.
\item System \eqref{gensys} with initial values 
$\mb{y}^0=(\mb{c}^0,v^0,\phi^0)\in \mathcal{C}_K\cap \Gamma_\alpha$ 
has a unique solution $\mb{y}(t)=(\mb{c}(t), v(t), \phi(t))$
for short times. For short times, $\phi(t)=\Phi(\mb{c}(t),\alpha)$.
\end{enumerate}
\end{proposition}
\begin{proof}
To construct the function $\Phi(\mb{c},\alpha)$, we 
use the implicit function theorem around 
$\mb{c}=\mb{c}^{\rm e}, \phi=0, \alpha=0$
on $I$.
Note that:
\begin{equation}
I(\phi=0,\mb{c}=\mb{c}^{\rm e}, \alpha=0)
=\ip{\mb{z}}{\mb{j}(\phi=0,\bm{\mu}=0)}_{\mathbb{R}^N}=0\label{Iat0}
\end{equation}
where we used the definition of $\bm{\mu}$ in the first 
equality and \eqref{noflux} in the second equality.
Take the derivative of $I$ with respect to $\phi$:
\begin{equation}
\PD{I}{\phi}=
\ip{\mb{z}}{\PD{\mb{j}}{\phi}+\PD{\mb{j}}{\bm{\mu}}\mb{z}}_{\mathbb{R}^N}
+\alpha \PD{p}{\phi}.
\end{equation}
In the above, $\mb{j}$ 
is viewed as a function of $\phi$ and $\bm{\mu}$ whereas 
$\mb{p}$ is viewed as a function of $\phi$ and $\mb{c}$. We have used 
the definition of $\bm{\mu}$ to obtain the second term in the above. 
We see that
\begin{equation}
\PD{I}{\phi}(\phi=0, \mb{c}=\mb{c}^{\rm e}, \alpha=0)
=\ip{\mb{z}}{\PD{\mb{j}}{\bm{\mu}}\mb{z}}_{\mathbb{R}^N}>0.\label{dIdphiat0}
\end{equation}
where we used \eqref{djkdphi}.
Positivity follows from \eqref{posdef}.
With \eqref{Iat0} and \eqref{dIdphiat0}, we can use
the implicit function theorem to obtain a $C^1$ function $\Phi$
satisfying
\begin{equation}
I(\Phi(\mb{c},\alpha),\mb{c},\alpha)=0, 
\; \Phi(\mb{c}^{\rm e}, 0)=0,
\end{equation}
in a neighborhood of $\mb{c}=\mb{c}^{\rm e}, \alpha=0$.
The rest of the proof is the same as that of Lemma 
\ref{gensolexist}.
\end{proof}

\subsection{Existence of Steady States and Asymptotic Stability}

A point $\mb{y}=(\mb{c},v,\phi)\in \Gamma$ is a steady state of \eqref{gensys}
if the right hand side of \eqref{ionsgen} and \eqref{volgen} is 
$0$ at that point. 
We have the following result on the existence of steady states.
This should be seen as a generalization of condition \eqref{32NaK}.
\begin{proposition}\label{existence}
Suppose $j_k$ and $j_{\rm w}$ satisfy \eqref{noflux}, \eqref{jwinc}
and \eqref{posdef}. Then, \eqref{gensys} has a steady state with 
$v>0$ for all sufficiently small $\alpha>0$ so long 
as the following condition is met:
\begin{equation}\label{ckejkpl}
\at{\ip{\mb{c}^{\rm e}}{\paren{\PD{\mb{j}}{\bm{\mu}}}^{-1}\mb{p}}_{\mathbb{R}^N}}{\phi=0, \bm{\mu}=\mb{0}}>0,
\end{equation}
where $(\partial \mb{j}/\partial \bm{\mu})^{-1}$ is the inverse of the Jacobian matrix $\partial \mb{j}/\partial \bm{\mu}$.
\end{proposition}
The idea behind this result is the following. System 
\eqref{gensys} possesses an obvious ``steady state'' when $\alpha=0$: 
$\mb{c}=\mb{c}^{\rm e}, \phi=0$ and $v=\infty$. If $\alpha$ is positive
but small, we expect this steady state to persist. In order for 
$v$ to be positive as $\alpha$ is perturbed, we need condition 
\eqref{ckejkpl}.
\begin{proof}[Proof of Proposition \ref{existence}]
Let $w=1/v$. Set the right hand side of \eqref{gensys} equal to $0$.
We have:
\begin{equation}\label{leftsideeqn}
\begin{split}
j_k(\phi,\bm{\mu})+\alpha p_k(\phi,\bm{\mu})&=0, \quad k=1,\cdots,N,\\
\sum_{k=1}^N z_kc_k +zAw&=0,\\
j_{\rm w}(\pi_{\rm w})&=0.
\end{split}
\end{equation}
View the above as an equation for $\mb{c}, \phi$ and $w$.
Note that $\mb{c}=\mb{c}^{\rm e}, \phi=0, w=0$ is a solution to the 
above system if $\alpha=0$. To apply the implicit function theorem, 
we show that the Jacobian matrix with 
respect to $\mb{c}, \phi, w$ is invertible at this point. 
This is equivalent to showing that the only solution to 
the following linear equation for 
$\widehat{\mb{c}}=(\widehat{c_1},\cdots,\widehat{c_N}), 
\widehat{\phi}$ and $\widehat{w}$ is the trivial one. 
\begin{align}
\sum_{l=1}^N 
\at{\PD{j_k}{\mu_l}}{\phi=0,\bm{\mu}=\mb{0}}
\paren{\frac{\widehat{c_l}}{c_l^{\rm e}}+z_l\widehat{\phi}}
&=0, \; k=1,\cdots, N,\label{jmulin}\\
\sum_{k=1}^N z_k\widehat{c_k} +zA\widehat{w}&=0,\label{ENlin}\\
\at{\PD{j_{\rm w}}{\pi_{\rm w}}}{\pi_{\rm w}=0}\paren{\sum_{k=1}^N \widehat{c_k}+A\widehat{w}}&=0,\label{osmlin}
\end{align}
where we used \eqref{noflux} (and its consequence \eqref{djkdphi})
to obtain \eqref{jmulin}. 
Equation \eqref{jmulin} together with condition \eqref{posdef} and
\eqref{osmlin} together with condition \eqref{jwinc} gives:
\begin{equation}\label{muklin}
\frac{\widehat{c_k}}{c_k^{\rm e}}+z_k\widehat{\phi}=0,\quad 
\sum_{k=1}^N \widehat{c_k}+A\widehat{w}=0
\end{equation}
where the first equation is holds for all $k$.
Using \eqref{muklin} to eliminate $\widehat{w}$ and 
$\widehat{c_k}$ from \eqref{ENlin}, we have:
\begin{equation}
\sum_{k=1}^N (z_k-z)z_k c_k^{\rm e}\widehat{\phi}=
\sum_{k=1}^N z_k^2 c_k^{\rm e}\widehat{\phi}=0
\end{equation}
where we used \eqref{ENgen} in the first equality. Since
$c_k^{\rm e}$ is positive and at least one of $z_k\neq 0$,
we see that $\widehat{\phi}=0$. From \eqref{muklin}, we see that 
$c_k=0$ for all $k$ and $\widehat{w}=0$ since $A>0$.
We can thus invoke the implicit function theorem to conclude 
that we have a solution $\mb{c}(\alpha), \phi(\alpha)$ and $w(\alpha)$
to \eqref{leftsideeqn} when $\alpha$ is close to $0$. To ensure that $w$ (or 
equivalently, $v$) is positive for small $\alpha>0$, we compute:
\begin{equation}\label{dwdalpha}
\at{\D{w}{\alpha}}{\alpha=0}=A^{-1}
\at{\ip{\mb{c}^{\rm e}}{\paren{\PD{\mb{j}}{\bm{\mu}}}^{-1}\mb{p}}_{\mathbb{R}^N}}{\phi=0, \bm{\mu}=\mb{0}}.
\end{equation}
Since $w=0$ at $\alpha=0$, condition \eqref{ckejkpl} will ensure that 
the $v$ is positive for $\alpha$ small and positive. 
\end{proof}
We may also compute $d\phi/d\alpha$:
\begin{equation}
\begin{split}
\at{\D{\phi}{\alpha}}{\alpha=0}&=
\paren{\sum_{k=1}^N z_k^2c_k^{\rm e}}^{-1}
\at{\ip{(z\mb{c}^{\rm e}-\mb{z}_c)}{\paren{\PD{\mb{j}}{\bm{\mu}}}^{-1}\mb{p}}_{\mathbb{R}^N}}{\phi=0, \bm{\mu}=\mb{0}},\\
\mb{z}_c&=(z_1c_1^{\rm e},\cdots,z_Nc_N^{\rm e}).
\end{split}
\end{equation}
Given \eqref{ckejkpl}, this shows that the sign of $\phi$ is the same 
as the sign of $z$ if $\abs{z}$ is large enough. In physiological situations, 
$z$ is large and negative, and thus we will have a negative membrane potential.

Condition \eqref{ckejkpl} applied to \eqref{NaKCl} yields:
\begin{equation}\label{32NaK2}
\frac{3[{\rm Na}^+]_{\rm e}}{g_{\rm Na}}-\frac{2[{\rm K}^+]_{\rm e}}{g_{\rm K}}>0,
\end{equation}
thus reproducing condition \eqref{32NaK}. 
It is interesting that we do not have any restriction on $z$ for this 
to be true. Given \eqref{dwdalpha}, the cell volume will be small
if \eqref{ckejkpl} is large. Expression \eqref{32NaK2} is
indeed large: for a typical cell, $[{\rm Na}^+]_{\rm e}\gg [{\rm K}^+]_{\rm e}$
and $g_{\rm Na}\ll g_{\rm K}$.

We now turn to the question of stability of steady states.
\begin{definition}
Suppose $\mb{y}^*=(\mb{c}^*,v^*,\phi^*)\in \Gamma$ is a steady state 
of \eqref{gensys}. 
Let $\mathcal{B}_r$ be the open ball of radius $r$ centered at $\mb{y}^*$. 
The steady state $\mb{y}^*$ is stable if the following is 
true. For any small enough $\epsilon>0$, there exists a $\delta>0$ with the 
following property. Choose any 
$\mb{y}^0\in \mathcal{B}_\delta\cap \Gamma$. Then, 
the solution(s) $\mb{y}(t)$ 
to system \eqref{gensys} with initial value 
$\mb{y}^0$ is defined for all positive time and satisfies
$\mb{y}(t)\in \mathcal{B}_\epsilon$.
The steady state is asymptotically stable if it is stable and 
all solutions with initial data $\mb{y}^0\in \Gamma$ 
sufficiently close to $\mb{y}^*$ 
approach $\mb{y}^*$ as $t\to \infty$.
\end{definition} 
Since any solution lies on $\Gamma$, we may replace $\mathcal{B}_\epsilon$
with $\mathcal{B}_\epsilon\cap \Gamma$.
The only difference between the usual definition of stability and 
the one above is that the initial data must lie on $\Gamma$. 
If $\partial{I}/\partial{\phi}\neq 0$ at $\mb{y}^*=(\mb{c}^*,v^*,\phi^*)$
and $z\neq 0$, then, by the remark after the proof of Lemma \ref{gensolexist}, 
system \eqref{gensys} can locally be written 
as an ODE for $\mb{c}$ only. The above definition of (asymptotic) stability
is then equivalent to the (asymptotic) stability of $\mb{c}^*$ for 
this ODE system. 

Let $\mb{c}^*=(c_1^{*}, \cdots,c_N^*), v^*, \phi^*$ be 
a steady state of \eqref{gensys}. 
Define the following quantities:
\begin{equation}
\begin{split}
\wh{\phi}&=\phi-\phi^*,\\
\wh{\bm{\gamma}}&=(\wh{\gamma}_1,\cdots,\wh{\gamma}_N)^T, \; 
\wh{\gamma}_k=\gamma_k-\gamma_k^*=\ln\paren{\frac{c_k}{c_k^*}}\\
\wh{\bm{\mu}}&=(\wh{\mu}_1,\cdots,\wh{\mu}_N)^T, \; 
\wh{\mu}_k=\mu_k-\mu_k^*\equiv \wh{\gamma}_k+z_k\wh{\phi},\\
\wh{\pi}_{\rm w}&=\pi_{\rm w}-\pi_{\rm w}^*
\equiv \sum_{k=1}^N c_k^*+\frac{A}{v^*}-\paren{\sum_{k=1}^N c_k+\frac{A}{v}}.
\end{split}
\end{equation}
If $j_{\rm w}$ satisfies \eqref{jwgencond}, $j_{\rm w}(\pi_{\rm w})=0$ if and only if 
$\pi_{\rm w}=0$ and thus 
$\pi_{\rm w}^*=0$. In this case, $\wh{\pi}_{\rm w}=\pi_{\rm w}$.

Let $\wh{G}$ be free energy with respect to the steady 
state defined in \eqref{hatG}.   
We have the following analogue of Proposition \ref{FEE} or Lemma \ref{FEq}.
\begin{lemma}\label{relative}
Suppose $\mb{c}^*=(c_1^{*}, \cdots,c_N^*)^T, \phi^*, v^*$ is 
a steady state of \eqref{gensys}. Then, we have:
\begin{equation}
\D{\wh{G}}{t}
=-\sum_{k=1}^N \wh{\mu}_k\paren{\wh{j}_k+\alpha\wh{p}_k}
-\wh{\pi}_{\rm w}j_{\rm w}
\end{equation}
where $\wh{j}_k=j_k-j_k^*, \; \wh{p}_k=p_k-p_k^*$ and 
$j_k^*,\; p_k^*$ are the passive and active fluxes evaluated 
at the steady state.
\end{lemma}
\begin{proof}
Rewrite \eqref{ionsgen} as follows:
\begin{equation}
\PD{(vc_k)}{t}=-(j_k+\alpha p_k)+\paren{j_k^*+\alpha p_k^*}=
-\sum_{k=1}^N \wh{\mu}_k\paren{\wh{j}_k+\alpha\wh{p}_k},
\end{equation}
where we used $j_k^*+\alpha p_k^*=0$.
Note that, $j_{\rm w}^*$, the water flux at steady state,  is 
equal to $0$. Thus $\wh{j}_{\rm w}
\equiv j_{\rm w}-j_{\rm w}^*=j_{\rm w}$.
The rest of 
the proof is the same as Proposition \ref{FEE}.
\end{proof}
If we apply the above lemma to system \eqref{linsys}, this 
is nothing other than \eqref{Ghateq}.
The next Lemma gives us a sufficient condition 
for asymptotic stability in terms of $\wh{G}$.
\begin{lemma}\label{stabcond}
Let $\mb{y}^*=(\mb{c}^*,v^*,\phi^*), \;\mb{c}^*=(c_1^{*}, \cdots,c_N^*)^T$ be
a steady state of \eqref{gensys}.
Suppose there is neighborhood $\mathcal{U}\subset\mathbb{R}^{N+2}$ of 
$\mb{y}^*$ such that $\mathcal{U}\cap \Gamma$ is an $N$-dimensional 
submanifold in which any point $\mb{y}
\in\mathcal{U}\cap \Gamma$ can be written as $\mb{y}=(\mb{c},v,\Phi(\mb{c}))$
where $\Phi$ is a $C^1$ function of $\mb{c}$.
Suppose any solution $\mb{y}(t)=(\mb{c}(t),v(t),\phi(t))$ in $\mathcal{U}$
(or equivalently, in $\mathcal{U}\cap \Gamma$)
satisfies:
\begin{equation}\label{vsigdec}
\D{\wh{G}}{t}\leq 
-K_*\paren{\abs{\wh{\bm{\mu}}}^2+\abs{\wh{\pi}_{\rm w}}^2}
\end{equation}
for some positive constant $K_*$. 
Then the steady state is asymptotically stable 
and the approach to the steady state is exponential in time. 
\end{lemma}
\begin{proof}
As in the proof of Proposition \ref{linmain}, 
we will find it convenient to 
use the variables $\mb{a}=(a_1,\cdots,a_N)^T=v\mb{c}, v$ and $\phi$ rather 
than $\mb{c}, v$ and $\phi$. 
We shall continue to use the symbols $\mathcal{U}, \Gamma$
to denote the corresponding sets in the new coordinates. 
View $\wh{G}$ as a function of $\mb{a}$ and $v$. Note first that:
\begin{equation}
\begin{split}
\wh{G}(\mb{a}^*,v^*)&=0,\\
\quad \at{\PD{\wh{G}}{a_k}}{\mb{a}=\mb{a}^*,v=v^*}&=
\at{\paren{\gamma_k-\gamma_k^*}}{\mb{c}=\mb{c}^*}=0,\\
\at{\PD{\wh{G}}{v}}{\mb{a}=\mb{a}^*,v=v^*}
&=\at{\pi_{\rm w}}{\mb{c}=\mb{c}^*, v=v^*}=0,
\end{split}
\end{equation}
where $\mb{a}^*=(a_1^*,\cdots,a_N^*)^T =v^*\mb{c}^*$.
By Lemma \ref{propG}, $\wh{G}$
is a globally convex function on $(\mb{a},v)\in \mathbb{R}_+^{N+1}$.
The point $(\mb{a},v)=(\mb{a}^*,v^*)$ is thus the global minimizer 
of $\wh{G}$. The positive definiteness of 
the Hessian matrix of $\wh{G}$
implies that there is a neighborhood 
$\mathcal{N}\subset\mathbb{R}^{N+1}$ of $(\mb{a}^*,v^*)$ where 
\begin{equation}\label{GCvequiv}
K_G^{-1}\paren{\abs{\mb{a}-\mb{a}^*}^2+\abs{v-v^*}^2}
\leq \wh{G}(\mb{a},v)\leq 
K_G \paren{\abs{\mb{a}-\mb{a}^*}^2+\abs{v-v^*}^2}
\end{equation}
for some positive constant $K_G$. 

Now, consider $\mathbb{R}^{N+2}$ with the coordinates 
$(\mb{a},v,\phi)$. Define $
Q=\sum_{k=1}^N z_k(a_k-a_k^*)$.
This $Q$ is the same as the $Q$ in \eqref{Qconstraint}
except that it is written in terms of $\mb{a}$.
We claim that, 
in the vicinity of $(\mb{a}^*,v^*,\phi^*)$ in $\mathbb{R}^{N+2}$, 
the set of variables $(\widehat{\bm{\mu}},\wh{\pi}_{\rm w}, Q)$ defines 
a coordinate system. It is easily seen that
the variables $(\mb{c},v,\phi)$ defines a coordinate system.
We thus consider the coordinate change from 
$(\widehat{\bm{\mu}},\wh{\pi}_{\rm w}, Q)$ to $(\mb{c},v,\phi)$. 
The Jacobian matrix between these two sets of variables 
at $(\mb{c},v,\phi)=(\mb{c}^*,v^*,\phi^*)$ is non-singular. 
This computation is almost the same as the computation 
in the proof of Proposition \ref{existence},
so we omit the details. The claim follows by the implicit 
function theorem. There is therefore a neighborhood 
$\mathcal{V}\subset\mathbb{R}^{N+2}$ of $(\mb{a}^*,v^*,\phi^*)$
in which the following inequality holds:
\begin{equation}
K_\mu\paren{\abs{\mb{a}-\mb{a}^*}^2+\abs{v-v^*}^2+\abs{\phi-\phi^{*}}^2}
\leq \abs{\widehat{\bm{\mu}}}^2+\abs{\wh{\pi}_{\rm w}}^2+\abs{Q}^2
\end{equation}
where $K_\mu$ is a positive constant.
Any solution to \eqref{gensys} satisfies 
$Q=0$ (see \eqref{Qconstraint}). Thus, we have:
\begin{equation}\label{Cvmupi}
K_\mu\paren{\abs{\mb{a}-\mb{a}^*}^2+\abs{v-v^*}^2}
\leq \abs{\widehat{\bm{\mu}}}^2+\abs{\pi_{\rm w}}^2
\end{equation}
for any solution in $\mathcal{V}$.

Choose a neighborhood $\mathcal{M}\subset\mathbb{R}^{N+1}$ of $(\mb{a}^*,v^*)$
such that $(\mb{a},v,\Phi(\mb{a}/v))\in \mathcal{U}\cap \mathcal{V}$
for all $(\mb{a},v)\in \mathcal{M}$. Consider the following set:
\begin{equation}
\mathcal{G}=\lbrace (\mb{a},v)\in \mathbb{R}^{N+1}\; |\; 
\wh{G}(\mb{a},v)<M_G, M_G>0\rbrace
\end{equation}
Since $\wh{G}(\mb{a},v)$ is a convex function such that $G(\mb{a}^*,v^*)=0$, 
$\mathcal{G}$ is an open neighborhood of 
$(\mb{a}^*,v^*)$, and can be made arbitrarily small by making $M_G$ small.
Take $M_G$ so small that $\mathcal{G}\subset \mathcal{N}\cap \mathcal{M}$.
Consider the following open neighborhood of $(\mb{a}^*,v^*,\phi^*)$:
\begin{equation}
\mathcal{W}=\lbrace \mb{y}=(\mb{a},v,\phi)\in \mathbb{R}^{N+2} \; | \; 
(\mb{a},v)\in \mathcal{G}, \mb{y}\in \mathcal{U}\cap \mathcal{V}\rbrace.
\end{equation}
Any solution in $\mathcal{W}$, by definition, belongs to 
$\mathcal{W}\cap \Gamma$. For any such solution, we have: 
\begin{equation}
\begin{split}
\D{\wh{G}}{t}&\leq -K_*(\abs{\widehat{\bm{\mu}}}^2+\abs{\wh{\pi}_{\rm w}}^2)\\
&\leq -K_*K_\mu\paren{\abs{\mb{a}-\mb{a}^*}^2+\abs{v-v^*}^2}
\leq -\frac{K_*K_\mu}{K_G}\wh{G}
\end{split}
\end{equation}
where we used \eqref{vsigdec} in the first inequality, \eqref{Cvmupi}
in the second inequality and \eqref{GCvequiv} in the third inequality.
Solving the above differential inequality, we have:
\begin{equation}
\wh{G}\leq M_G\exp(-K_*K_\mu t/K_G), \;  t\geq 0.
\end{equation}
We thus see that $\mathcal{W}\cap \Gamma$ is a positively invariant 
set, and thus, all solutions starting from $\mathcal{W}\cap \Gamma$
are defined for all time. By \eqref{GCvequiv}, $(\mb{a}(t),v(t))$ approaches 
$(\mb{a}^*,v^*)$ exponentially in time. 
Since $(\mathcal{W}\cap\Gamma)\subset(\mathcal{U}\cap \Gamma)$, 
$\phi=\Phi(\mb{a}(t)/v(t))$. Since $\Phi$ is a $C^1$ function,
$\phi(t)$ also approaches $\phi^*$ exponentially in time. 
\end{proof}

We are now ready to state the main result of this Section.
\begin{theorem}\label{genmain}
Suppose $j_k$ and $j_{\rm w}$ satisfy \eqref{noflux}, \eqref{posdef}
and \eqref{jwinc} and $j_k, p_k$ and $c_k^{\rm e}$ satisfy \eqref{ckejkpl}. 
For all sufficiently small $\alpha>0$, 
the steady states found in Proposition \ref{existence} are asymptotically 
stable. The approach to steady state is exponential in time.
\end{theorem}

In Proposition \ref{linmain}, we used the symmetry condition of \eqref{posdef}
to show that the eigenvalues of the linearized matrix around steady
state are all real. Here, we cannot prove such a statement.
In fact, the proof to follow
goes through even if we assume \eqref{weakposdef} instead of 
\eqref{posdef}.

We also point out that, unlike Proposition \ref{linmain} or Theorem \ref{globalstab}, 
we can only draw conclusions when the pump rate is small ($\alpha$ is small). 
One may wonder whether it may be possible to generalize 
Theorem \ref{genmain} to the case when the pump rate is 
not necessarily small. 
For this, one would clearly need a condition 
stronger than \eqref{jwinc} or \eqref{posdef}. 
One natural idea is to require that 
$j_k$ and $j_{\rm w}$ satisfy \eqref{jwinc} and \eqref{posdef}
not only at $\pi_{\rm w}=0$ and $\bm{\mu}=0$ but at any 
arbitrary value of $\pi_{\rm w}$ and $\bm{\mu}$:
\begin{equation}\label{jwjkstrong}
\begin{split}
\PD{j_{\rm w}}{\pi_{\rm w}}(\pi_{\rm w})&>0 \text{ for all } \pi_{\rm w},\\
\PD{\mb{j}}{\bm{\mu}}(\phi,\bm{\mu})&
\text{ is a positive definite matrix for all } \phi
\text{ and } \bm{\mu}.
\end{split}
\end{equation}
This stronger condition is indeed satisfied by the Goldman 
equation \eqref{jkGHK} (and trivially by \eqref{jkL}). 
Let us assume the pump rates $\alpha p_k$ are constant.
A natural conjecture may be that if $j_{\rm w}$ and $j_k$
satisfy \eqref{jwjkstrong}, any steady state of \eqref{gensys} 
is stable (whether or not the pump rate is small). 
This statement is true {\em 
provided that $j_k$ is only a function of $\bm{\mu}$ and not 
a function of $\phi$}. 
In the case of \eqref{jkL} or \eqref{linsys}, 
$j_k$ is indeed only a function of $\bm{\mu}$.
Unfortunately, \eqref{jkGHK} is a function of 
both $\bm{\mu}$ and $\phi$. 
The danger when $j_k$ depends on $\phi$ independently of $\bm{\mu}$
is that $\partial j_k/\partial \phi$ may 
adversely affect the stability properties imparted by condition 
\eqref{jwjkstrong}. 

When the pump strength is small ($\alpha$ small), 
$\partial j_k/\partial \phi$ is small 
provided $\alpha$ is small thanks to condition \eqref{djkdphi}. 
This is one of the key observations 
that we will use in the proof to follow.

\begin{proof}[Proof of Theorem \ref{genmain}]
Let $\mb{y}^*=(\mb{c}^*,v^*,\phi^*)$ be the steady state found 
in Proposition \ref{existence}.
As $\alpha>0$ is made small, 
$\mb{y}^*\in \mathcal{C}_K\cap \Gamma$ 
defined in Proposition \ref{wellposed}. 
We shall henceforth assume that $\mb{y}^*\in \mathcal{C}_K\cap \Gamma$.

By Proposition \ref{wellposed},  
Lemma \ref{relative} and Lemma \ref{stabcond}, it is sufficient
to show that, for sufficiently small $\alpha$, 
there is a neighborhood 
$\mathcal{U}\subset \mathbb{R}^{N+2}$
of $(\mb{c}^*,v^*,\phi^*)$ such that any $\mb{y}=(\mb{c},v,\phi)\in 
\mathcal{U}\cap \Gamma$ satisfies the following inequality:
\begin{equation}\label{suffcond}
K_*\paren{\abs{\wh{\bm{\mu}}}^2+\abs{\wh{\pi}_{\rm w}}^2}
\leq \sum_{k=1}^N \wh{\mu}_k\paren{\wh{j}_k+\alpha\wh{p}_k}
+\wh{\pi}_{\rm w}j_{\rm w}\equiv J
\end{equation}
for some $K_*>0$. The neighborhood $\mathcal{U}$ may depend on $\alpha$.

Recall from the proof of Lemma \ref{stabcond}
that $(\widehat{\bm{\mu}},\pi_{\rm w}, Q)$ defines a coordinate 
system in the vicinity $\mathcal{N}$ of $(\mb{c}^*,v^*,\phi^*)$. 
Note that the point $(\mb{c}^*,v^*,\phi^*)$ is the origin in 
the coordinate system $(\widehat{\bm{\mu}},\pi_{\rm w}, Q)$.
Let:
\begin{equation}
\mathcal{D}_r=\lbrace (\mb{c},v,\phi)\in \mathcal{N} \; | \; 
\abs{\widehat{\bm{\mu}}}^2+\wh{\pi}_{\rm w}^2+Q^2<r^2, r>0\rbrace
\end{equation}
and take $r$ small enough so that $\mathcal{D}_r\subset \mathcal{C}_K$.
Define the set:
\begin{equation}
\Gamma_Q=\lbrace \mb{y}=(\mb{c},v,\phi)\in \mathbb{R}_+^N\times 
\mathbb{R}_+\times \mathbb{R} \; | \; Q(\mb{c},v)=0\rbrace.
\end{equation}
Since $\Gamma\subset \Gamma_Q$, it is clearly sufficient if we 
can find a small enough $r>0$ such that \eqref{suffcond} holds for any 
point in $\mathcal{D}_r\cap \Gamma_Q$. 
 
We can write $\wh{\phi}=\phi-\phi^*$
as a function of $\widehat{\bm{\mu}},\pi_{\rm w}$ and $Q$ in $\mathcal{D}_r$.
In particular, on $\mathcal{D}_r\cap\Gamma_Q$, $\wh{\phi}$ is a
function of $\wh{\bm{\mu}}$ and $\wh{\pi}_{\rm w}$ only. 
Call this function $\wh{\phi}=\varphi(\wh{\bm{\mu}},\wh{\pi}_{\rm w})$.
The function $\varphi$ satisfies:
\begin{equation}
0=\sum_{k=1}^N (z_k-z)c_k^*\paren{\exp
\paren{\widehat{\mu}_k-z_k\varphi}-1}
-z\wh{\pi}_{\rm w}.\label{varphidef}
\end{equation}
This equation is obtained by expressing $c_k$ and $v$ 
in terms of $\widehat{\mu}_k, \wh{\pi}_{\rm w}$ and $\wh{\phi}$ 
and substituting this into $Q=0$.

Now, take any point 
$\mb{w}=
(\wh{\bm{\mu}}^1,\wh{\pi}_{\rm w}^1, 0)=
\in\mathcal{D}_r\cap \Gamma_Q, \wh{\bm{\mu}}^1=
(\wh{\mu}_1^1,\cdots,\wh{\mu}_N^1)^T
$ where we have expressed the point
$\mb{w}$ using the $(\wh{\bm{\mu}},\wh{\pi}_{\rm w}, Q)$ coordinate 
system. 
Let us compute the right hand side of \eqref{suffcond} at this point.
\begin{equation}\label{R}
J(\wh{\bm{\mu}}^1,\wh{\pi}_{\rm w}^1)=\sum_{k=1}^N \wh{\mu}_k^1
\paren{\wh{j}_k(\varphi(\wh{\bm{\mu}}^1,\wh{\pi}_{\rm w}^1),
\wh{\bm{\mu}}^1)
+\alpha\wh{p}_k(\varphi(\wh{\bm{\mu}}^1,\wh{\pi}_{\rm w}^1),
\wh{\bm{\mu}}^1)}
+\wh{\pi}_{\rm w}^1j_{\rm w}(\wh{\pi}_{\rm w}^1)
\end{equation}
where we took $\wh{j}_k,\;\wh{p}_k$ as functions of 
$\wh{\phi}, \; \wh{\bm{\mu}}$
and $\wh{j}_{\rm w}$ as a function of $\wh{\pi}_{\rm w}$.
For $\wh{j}_k$, we have:
\begin{equation}
\begin{split}
\wh{j}_k(\varphi(\wh{\bm{\mu}}^1,\wh{\pi}_{\rm w}^1),\wh{\bm{\mu}}^1)=&\int_0^1 \D{}{s}\wh{j}_k\paren{\varphi(s\wh{\bm{\mu}}^1,s\wh{\pi}_{\rm w}),s\wh{\bm{\mu}}^1}ds\\
=&\sum_{l=1}^N\wh{\mu}_l^1
\int_0^1 \paren{\paren{\PD{\wh{j}_k}{\wh{\phi}}\PD{\varphi}{\mu_l}+
\PD{\wh{j}_k}{\wh{\mu}_l}}(s\wh{\bm{\mu}}^1,s\wh{\pi}_{\rm w})}ds\\
&+\wh{\pi}_{\rm w}^1\int_0^1\paren{
\PD{\wh{j}_k}{\phi}\PD{\varphi}{\wh{\pi}_{\rm w}}
(s\wh{\bm{\mu}}^1,s\wh{\pi}_{\rm w})}ds.
\end{split}
\end{equation}
Performing a similar calculation for $\wh{p}_k$ and $\wh{j}_{\rm w}$
and substituting this back into \eqref{R}, we obtain the following 
expression.
\begin{equation}
\begin{split}
J(\bm{\psi})&=\ip{\bm{\psi}}{P\bm{\psi}}_{\mathbb{R}^{N+1}}, \quad
\bm{\psi}=(\wh{\bm{\mu}}^1,\wh{\pi}_{\rm w}^1),\\
P&=\int_0^1(L+B+C)(s\wh{\bm{\mu}}^1,s\wh{\pi}_{\rm w}^1)ds,
\end{split}
\end{equation}
where  
$L,B$ and $C$ are $(N+1)\times (N+1)$ matrix-valued functions
defined on $\mathcal{D}_r\cap \Gamma_Q$, given as follows.
Let $L_{kl}, B_{kl}$ and $C_{kl}$
be the $kl$ entries of these matrices.
\begin{align}
L_{kl}&=\begin{cases}
\partial{\wh{j}_k}/\partial{\wh{\mu}_l} &\text{ if } 1\leq k,l\leq N,\\
\partial{\wh{j}_{\rm w}}/\partial{\wh{\pi}_w} &\text{ if } k=l=N+1,\\
0 & \text{ otherwise},
\end{cases}\\
B_{kl}&=\begin{cases}
(\partial{\wh{j}_k}/\partial{\wh{\phi}})
(\partial \varphi/\partial \wh{\mu}_l)
&\text{ if } 1\leq k,l\leq N,\\
(\partial{\wh{j}_k}/\partial{\wh{\phi}})
(\partial \varphi/\partial \wh{\pi}_{\rm w})
&\text{ if } 1\leq k\leq N, \; l=N+1,\\
0 &\text{ otherwise},
\end{cases}\\
C_{kl}&=\begin{cases}
\alpha(\partial{\wh{p}_k}/\partial{\wh{\mu}_l}+
(\partial{\wh{p}_k}/\partial{\wh{\phi}})
(\partial \varphi/\partial \wh{\mu}_l))
&\text{ if } 1\leq k,l\leq N,\\
\alpha (\partial{\wh{p}_k}/\partial{\wh{\phi}})
(\partial \varphi/\partial \wh{\pi}_{\rm w})
&\text{ if } 1\leq k\leq N, \; l=N+1,\\
0 &\text{ otherwise}.
\end{cases}
\end{align}
To show that \eqref{suffcond} is valid in $\mathcal{D}_r\cap \Gamma_Q$, 
it is sufficient to show that $L+B+C$ is positive definite in 
$\mathcal{D}_r\cap \Gamma_Q$ in the sense that:
\begin{equation}\label{ABC}
\ip{\mb{x}}{(L+B+C)\mb{x}}_{\mathbb{R}^{N+1}}\geq K_*\abs{\mb{x}}^2
\end{equation}
for any $\mb{x}\in \mathbb{R}^{N+1}$ with a constant $K_*>0$ that 
does not depend on the point in $\mathcal{D}_r\cap \Gamma_Q$.
Since $L,B$ and $C$ are continuous functions on $\mathcal{D}_r\cap \Gamma_Q$
and we may take $r$ as small as we want, all 
we have to show is that \eqref{ABC} holds at the origin, 
$(\wh{\bm{\mu}},\wh{\pi}_{\rm w},Q)=(\mb{0},0,0)$, or equivalently, 
at the steady state.

Let $L^*,B^*$ and $C^*$ be the evaluation of the three matrices 
at steady state. Since the steady state is a function of $\alpha$, 
$L^*,B^*$ and $C^*$ are functions of $\alpha$. 
We will show that \eqref{ABC} holds for sufficiently small $\alpha>0$.

First, let us examine the behavior of $\bm{\mu}^*,\pi_{\rm w}^*,\phi^*$
as a function of $\alpha$. By Proposition \ref{existence}, 
$\mb{c}^*,\phi^*$ are $C^1$ functions of $\alpha$ that 
approach $\mb{c}^{\rm e}, 0$ respectively as $\alpha \to 0$. 
Therefore, $\bm{\mu}^*$ is a $C^1$ function of $\alpha$ 
that approaches $\mb{0}$ as $\alpha \to 0$. 
It is clear that $\pi_{\rm w}^*=0$ for any $\alpha$.

The $kl$ entry of the matrix $L^*$ is given by:
\begin{equation}
L_{kl}^*=\begin{cases}
\at{\paren{\partial{j_k}/\partial{\mu_l}}}{\phi=\phi^*,\bm{\mu}=\bm{\mu^*}} &\text{ if } 1\leq k,l\leq N,\\
\at{\paren{\partial{j_{\rm w}}/\partial{\pi_{\rm w}}}}{\pi_{\rm w}=0} 
&\text{ if } k=l=N+1,\\
0 & \text{ otherwise}.
\end{cases}
\end{equation}
Since $\bm{\mu}^* \to \mb{0}$ and $\phi^* \to 0$ as $\alpha \to 0$,
given \eqref{posdef} and \eqref{jwinc}, there is a constant $K_L>0$
such that:
\begin{equation}\label{KA}
\ip{\mb{x}}{L^*\mb{x}}_{\mathbb{R}^{N+1}}\geq K_L\abs{\mb{x}}^2
\end{equation}
for sufficiently small $\alpha>0$.

To examine $B^*$ and $C^*$, let us first compute 
$\partial{\varphi}/\partial{\wh{\mu}_l}$ and 
$\partial{\varphi}/\partial{\wh{\pi}_{\rm w}}$.
This can be computed by taking the partial derivatives of \eqref{varphidef}:
\begin{equation}
\begin{split}
\at{\PD{\varphi}{\wh{\mu}_l}}{\wh{\bm{\mu}}=\mb{0}, \wh{\pi}_{\rm w}=0}
&=(z_l-z)c_l^*\paren{\sum_{k=1}^Nz_k^2c_k^*+z^2\frac{A}{v^*}}^{-1},\\
\at{\PD{\varphi}{\wh{\pi}_{\rm w}}}{\wh{\bm{\mu}}=\mb{0}, \wh{\pi}_{\rm w}=0}
&=-z\paren{\sum_{k=1}^Nz_k^2c_k^*+z^2\frac{A}{v^*}}^{-1}.
\end{split}
\end{equation}
Since $c_k^*\to c_k^e$ and $1/v^*\to 0$ as $\alpha \to 0$, 
we see that both of the above are bounded (and in fact 
has a definite limit) as $\alpha \to 0$.

Let us examine $B^*$. For $1\leq k,l\leq N$, the $kl$ entry of 
the matrix $B^*$ is given by:
\begin{equation}
B^*_{kl}=(z_l-z)c_l^*\paren{\sum_{k=1}^Nz_k^2c_k^*+z^2\frac{A}{v^*}}^{-1}
\at{\paren{\PD{j_k}{\phi}}}{\phi=\phi^*,\bm{\mu}=\bm{\mu}^*}.
\end{equation}
Recall that $\partial{j_k}/\partial{\phi}=0$ at $\bm{\mu}=\mb{0}$ from 
\eqref{djkdphi}.
Since $\phi^* \to 0$ and $\bm{\mu}^* \to \mb{0}$ as $\alpha \to 0$, 
by \eqref{noflux}, we see that $B_{kl}^* \to 0$ as $\alpha \to 0$.
Likewise, $B_{kl}^* \to 0$ as $\alpha \to 0$ when $1\leq k\leq N$ and 
$l=N+1$. 

Now, consider $C^*_{kl}$, the $kl$ entries of the matrix $C^*$.
For $1\leq k,l\leq N$, we have: 
\begin{equation}
C^*_{kl}=\alpha\at{
\paren{\PD{p_k}{\mu_l}
+(z_l-z)c_l^*\paren{\sum_{k=1}^Nz_k^2c_k^*+z^2\frac{A}{v^*}}^{-1}
\PD{p_k}{\phi}}}{\phi=\phi^*,\bm{\mu}=\bm{\mu}^*}.
\end{equation}
Given that $\bm{\mu}^*\to \mb{0},\; \phi^* \to 0$ as $\alpha \to 0$,
the quantity inside the outer-most parentheses remains bounded 
as $\alpha \to 0$. Thus, $C^*_{kl} \to 0$ as $\alpha \to 0$. 
The same conclusion holds for $C^*_{kl},\; 1\leq k\leq N,\; l=N+1$.

Since $L^*$ satisfies \eqref{KA} from sufficiently small $\alpha$
and $B^*$ and $C^*$ both tend to the zero matrix
as $\alpha \to 0$, $L^*+B^*+C^*$ satisfies \eqref{ABC}
with $K_*=K_L/2$ for small enough $\alpha$.
\end{proof}

\section{Discussion}

In this paper, we presented what the author believes is 
the first analytical result on the stability of steady states of pump-leak 
models. 
In Section \ref{lin}, we studied the case in which 
the flux functions are linear in the chemical potential. 
In the proof of Proposition \ref{linmain}, we saw that 
the system can be seen as a gradient flow of a convex function.
This is nothing other than the relaxation law 
postulated in linear non-equilibrium thermodynamics,
$d\mb{X}/dt=-L\nabla_X G$,
where $\mb{X}$ is the vector of extensive variables, 
$L$ is the matrix of transport coefficients 
and $\nabla_X G$ is the gradient
of the free energy $G$ with respect to $\mb{X}$
\citep{katzir1965nonequilibrium,kjelstrup2008non}. 
There are two interesting points here. 
The first point is that this gradient flow is restricted 
to flow on a submanifold on which electroneutrality
holds. The electrostatic potential, as we 
discussed in the proof of Lemma \ref{propG} can then be seen
as a Lagrange multiplier of this gradient flow.
The second point 
is that we can find a suitable modification of the free 
energy ($\wt{G}$ or $\wh{G}$) so that our system is a 
gradient flow even in the presence of active currents.
We proved the following results.
If condition \eqref{fphimin}
is satisfied, there is a unique steady state that is 
globally asymptotically stable. If not, the cell volume 
tends to infinity as time $t\to \infty$.
The system is thus robust to external perturbations in the following sense.
Suppose the cell is subject to a change in 
extracellular concentration or pump rate that stays within 
the bounds of condition \eqref{fphimin}. The cell
will eventually approach the new global steady state. 
We saw that the same conclusions hold for a simple epithelial model, 
since it could be mapped to the single cell problem. 

In Section \ref{general}, we proved that 
steady states for pump-leak models are stable so long as 
the steady state is not too far away from thermodynamic equilibrium. 
The ``stable equilibrium state'' is 
the ``state'' at which all intracellular ionic concentrations $c_k$
are equal to the extracellular concentration $c_k^{\rm e}$ and the 
cell volume $v$ is infinite.
If the pumps work in the ``right direction'' (in the sense of 
Proposition \ref{existence})
$v$ can be made finite even with a small pump rate. 
A biophysical interpretation of Theorem \ref{genmain} is that  
if the pump rate is sufficiently 
small, the new steady state of finite volume 
is still close enough to thermodynamic equilibrium 
so that the steady state inherits the stability properties of the 
equilibrium state. 
It is interesting that stability 
of  thermodynamic equilibrium, 
which may be considered the ``dead'' state, confers stability 
to the ``live'' state.

The results of Section \ref{general}, 
though applicable to general pump-leak models,
only asserts the existence of at least 
one asymptotically stable steady state 
for small pump rates. 
To draw analytical conclusions at large pump rates
without the linearity assumption of Section \ref{lin}, 
it is likely that one would need to look at
special characteristics of specific pump-leak models. 
Our current study may thus be seen as complementing 
computational investigations of stability, in which 
one is not constrained to small pump rates 
\citep{weinstein1997dynamics}. We also point out that, 
for general pump-leak models, we cannot 
rule out the possibility of multiple steady states 
or of other non-trivial asymptotic behavior.
Indeed, \citep{weinstein1992analysis} reports
an instance in which there are two stable steady states 
in a non-electrolyte model of epithelial cell volume control.

Many epithelial models include effects not 
included in model \eqref{system} or \eqref{epi}. 
Of particular importance is the incorporation of  
acid-base reactions
\citep{weinstein1983nonequilibrium,strieter_volume-activated_1990}.
It is usually assumed that the acid-base reactions are sufficiently 
fast so that these reactions are in equilibrium. This gives 
rise to additional algebraic constraints, increasing the 
co-dimension of the differential algebraic system 
\citep{weinstein2002assessing,weinstein2004modeling,weinstein2009modeling}. 
It would be interesting to see whether the analysis of this 
paper can be extended to this case. 
A starting point for an analytical study of such models will 
likely be a free energy identity. A potential complication 
is that the algebraic constraints of acid-base reactions 
are not linear in the concentrations. This may pose difficulties 
in extending the global results of Section \ref{lin}. 
The author hopes to report on such an extension in future work.

Stability of steady states is just a starting point 
in the study of homeostatic control in epithelial systems. 
In \citep{weinstein2002assessing,weinstein2004modeling,weinstein2009modeling}, 
the authors go beyond stability to study 
the optimal control of homeostatic parameters by minimizing 
a quadratic cost function along a relaxation trajectory.  
We hope that our current study will
lead to new insights into such problems.

Free energy dissipation identities similar to \eqref{FE} are present 
in many models of soft-condensed matter physics 
\citep{doi1988theory,doi2009gel}, 
and can be used as a guiding principle in formulating 
models in dissipative systems \citep{eisenberg2010energy,mori_liu_eis}. 
The present work 
owes much of its inspiration to this body of work. 
We believe that there is much to be gained by a systematic 
application of these ideas to the study of physiological 
systems. We hope that this paper will be a starting point 
in this direction.

\noindent \textbf{Acknowledgments}\newline
\noindent This work was inspired by 
the numerous discussions the author had with Chun Liu and 
Robert S. Eisenberg at the IMA (Institute for Mathematics 
and its Applications). 
This work was supported by 
NSF grant DMS-0914963,
the Alfred P. Sloan Foundation and the McKnight Foundation.

\bibliographystyle{apalike}
\bibliography{mylib}

\end{document}